\documentclass[aip,jmp,superscriptaddress]{revtex4-1}

\usepackage{latexsym}
\usepackage[T1]{fontenc}
\usepackage[utf8x]{inputenc}
\usepackage{times}
\usepackage{amsmath}
\usepackage{amssymb}
\usepackage{amscd}
\usepackage{amsfonts}
\usepackage{amstext}
\usepackage{amsthm}
\usepackage{graphicx}
\usepackage[all]{xy}
\usepackage{float}
\usepackage{multirow}
\usepackage{color}
\newcommand{\hq }{{/\kern -.185em/}}
\def\dim{\operatorname{dim}}

\theoremstyle{plain}

\newtheorem{theorem} {Theorem}

\newtheorem{prop}[theorem]{Proposition}
\newtheorem{cor} [theorem]{Corollary}

\theoremstyle{definition}

\newtheorem{example} {Example}

\newtheorem{fact}{Fact}

\newcommand{\ket}[1]{\mbox{$| #1 \rangle$}}

\newcommand{\kb}[2]{\ensuremath{| #1 \rangle\!\langle #2 |}}

\begin{document}

\title{Geometry and topology of CC and CQ states}

\author{Micha{\l} Oszmaniec}
\affiliation{Center for Theoretical Physics, Polish Academy of
Sciences, Al. Lotnik\'ow 32/46, 02-668 Warszawa, Poland}

\author{Piotr Suwara}
\affiliation{College of Inter-Faculty Individual Studies in Mathematics
and Natural Science, University of Warsaw, ul. \.{Z}wirki i Wigury
93, 02-089 Warszawa, Poland}

\author{Adam Sawicki}
\affiliation{Center for Theoretical Physics, Polish Academy of
Sciences, Al. Lotnik\'ow 32/46, 02-668 Warszawa, Poland}
\affiliation{School of Mathematics, University of Bristol, University Walk,
Bristol BS8 1TW, U.K.}

\begin{abstract}
We show that mixed bipartite $CC$ and $CQ$ states are geometrically
and topologically distinguished in the space of states. They are characterized
by non-vanishing Euler-Poincar\'{e} characteristics on the topological
side and by the existence of symplectic structures
on the geometric side.
\end{abstract}


\maketitle


\section{Introduction}

The existence of quantum correlations for multipartite separable mixed
states can be regarded as one of the most interesting quantum information
discoveries of the last decade. In 2001 Ollivier and \.{Z}urek \cite{Discord Zurek}
(see also \cite{Discord Verdal}) introduced the notion of quantum
discord as a measure of the quantumness of correlation. Quantum discord
is always non-negative \cite{Nullity of quantum discord}. The states
with vanishing quantum discord are called {\it pointer states}. They
form the boundary between classical and quantum correlations \cite{Nullity of quantum discord}.
The bipartite pointer states can be identified with so-called classical-quantum,
$CQ$ states \cite{Nullity of quantum discord}. An important subclass
of CQ states are classical-classical, CC states which play an important
role in quantification of the quantum correlations \cite{Horo 1,Horo 2}
and were recently considered in the context of broadcasting scenarios
\cite{Korbicz CC and CQ,chusc}. It is known that both classes are
of measure zero in the space of density matrices \cite{Acin almost all}.

In this paper we focus on the symplecto-geometric and topological
characterizations of mixed bipartite $CC$ and $CQ$ states. In [\onlinecite{Symplectic geometry of entanglement}]
it was shown that pure separable states form the unique symplectic
orbit of the local unitary (LU) group action. All other LU action
orbits are non-symplectic. Moreover, the more non-symplectic is a
LU action orbit the more entangled are states belonging to it \cite{Symplectic geometry of entanglement,HKS12}.

As we show in the present work, the symplectic techniques can be also
applied to describe mixed states. Remarkably, the non-degenerate symplectic
structure is present on a generic local unitary orbit through $CC$
and $CQ$ states rather than separable states. More precisely, for
density matrices defined on ${\mathcal{H}=\mathbb{C}^{N_{1}}\otimes\mathbb{C}^{N_{2}}}$
we show that the closure of all symplectic ${SU(N_{1})\times SU(N_{2})}$-orbits
is exactly the set of $CC$ states. Similarly, the closure of all
symplectic ${SU(N_{1})\times I_{N_{2}}}$-orbits gives the set of $CQ$
states. This clearly indicates that symplecticity generically detects
non-quantum rather than non-entangled states. For pure states two
concepts overlap. 

We also provide the topological characterization of pure seprarable,
$CC$ and $CQ$ states. More precisely, for pure $L$-partite and
bipartite mixed states we study Euler-Poincar\'{e} characteristics,
$\chi$ on LU action orbits. Using Hopf-Samelson theorem \cite{Hopf Samelson}
we show that for pure states $\chi$ is non-zero exactly on the manifold
of separable states. Moreover, for bipartite mixed states, ${SU(N_{1})\times SU(N_{2})}$-orbit
has non-vanishing Euler-Poincar\'{e} characteristics if and only
if the states belonging to it are $CC$. Similar result is true for
${SU(N_{1})\times I_{N_{2}}}$-orbits and $CQ$ states. As a conclusion
separable, $CC$, and $CQ$ states are topologically distinguished.

The paper is organized as follows. In Section \ref{sec: coadjoint orbits}
we discuss the relevant geometric structures present on the manifold
of bipartite isospectral density matrices, $\mathcal{O}_{\rho}$.
The orbits of $SU(N_{1})\times SU(N_{2})$ and ${SU(N_{1})\times I_{N_{2}}}$
are in a natural way homogenous submanifolds of this manifold. In
Section \ref{sec:Restrictions-of-geometric} we discuss the restriction
of the geometric structures to arbitrary homogenous submanifolds of
$\mathcal{O}_{\rho}$. In Section \ref{sec:geometry of mixed states}
we show how symplectic and K\"{a}hler structures distinguish classes
of orbits trough $CC$ and $CQ$ states. The second part of the article
deals with the topological characterization of these orbits. Section
\ref{sec:topology} discusses the Hopf-Samelson theorem for calculation
of the Euler-Poincar\'{e} characteristic of homogenous spaces. In
subsequent Section \ref{sec:Topol_results} we compute the Euler-Poincar\'{e}
characteristic of orbits of pertinent groups through pure separable,
$CC$, and $CQ$ states and show that these are the only orbits with
non-zero Euler-Poincar\'{e} characteristic.

By Facts we always denote results that are known. We do not present
their proofs and refer the reader to the literature. On the other
hand, by Propositions and Corollaries we denote all new results. Their
proofs are included in the text.

\section{Geometric structures on the manifold of isospectral density matrices
\label{sec: coadjoint orbits}}

A bipartite density matrix is a non-negative, trace-one operator $\rho$
on ${\mathcal{H}=\mathbb{C}^{N_{1}}\otimes\mathbb{C}^{N_{2}}}$, i.e.
an ${N\times N}$ matrix, ${N=N_{1}N_{2}}$, whose spectrum ${\sigma(\rho)=\{p_{1},\ldots,p_{N}\}}$
consists of non-negative eigenvalues ${p_{i}\geq0}$ satisfying ${\sum_{i=1}^{N}\,p_{i}=1}$.
Two density matrices are {\it isospectral} if they have the same
spectra. In the following we discuss geometric structures present
on the set of isospectral density matrices. In particular, we show
that this set is a compact K\"{a}hler manifold, that is, there exist
mutually compatible symplectic, Riemannian and complex structures
on it.

Let $\rho_{0}$ be a diagonal bipartite density matrix. The density
matrices which are isospectral with $\rho_{0}$ form an adjoint orbit
through $\rho_{0}$, $\mathcal{O}_{\rho_{0}}$, of ${G=SU(N)}$ action

\[
\mathcal{O}_{\rho_{0}}=\left\{ \mathrm{Ad}_{g}(\rho_{0}):\, g\in G\right\} \,,\,\mathrm{Ad}_{g}(\rho_{0})=g\rho_{0}g^{-1}\,.
\]
The Lie algebra $\mathfrak{g}$ of group $G$, i.e. the space of anti-hermitian
${N\times N}$ traceless matrices is equipped with the $G$-invariant
inner product defined by

\begin{equation}
\left(X\,|\, Y\right)=-\mathrm{tr}(XY),\,\,\left(gX\,|\, gY\right)=\left(X\,|\, Y\right)\,,\, g\in G\,.\label{eq:inner}
\end{equation}
For ${\rho\in\mathcal{O}_{\rho_{0}}}$ let ${G_{\rho}=\{g\in G:\,\mathrm{Ad}_{g}(\rho)=\rho\}}$
be the stabilizer of $\rho$ and ${\mathfrak{g}_{\rho}=\{X\in\mathfrak{g}:\,[X,\rho]=0\}}$
its Lie algebra. The geometric structures we want to discuss are defined
on the tangent bundle of $\mathcal{O}_{\rho_{0}}$, ${T\mathcal{O}_{\rho_{0}}=\bigcup_{\rho\in\mathcal{O}_{\rho_{0}}}T_{\rho}\mathcal{O}_{\rho_{0}}}$.
Thus we first need to describe ${T_{\rho}\mathcal{O}_{\rho_{0}}}$,
the tangent space to $\mathcal{O}_{\rho_{0}}$ at any ${\rho\in\mathcal{O}_{\rho_{0}}}$.
To this end, for ${X\in\mathfrak{g}}$ consider the corresponding fundamental
vector field $\tilde{X}$

\begin{equation}
\tilde{X}_{\rho}=\left.\frac{d}{dt}\right|_{t=0}e^{tX}\rho e^{-tX}=[X,\rho]\,.\label{eq:tangent-vector}
\end{equation}
As the action of $G$ on $\mathcal{O}_{\rho_{0}}$ is transitive,
the fundamental vector fields at ${\rho\in\mathcal{O}_{\rho_{0}}}$
span ${T_{\rho}\mathcal{O}_{\rho_{0}}}$. Note that for ${X\in\mathfrak{g}_{\rho}}$
the corresponding fundamental vector field vanishes, ${\tilde{X}_{\rho}=0}$.
Therefore, the tangent space ${T_{\rho}\mathcal{O}_{\rho_{0}}}$ can
be identified with $\mathfrak{g}_{\rho}^{\bot}$, that is, with the
orthogonal complement with respect to the inner product (\ref{eq:inner})
of $\mathfrak{g}_{\rho}$, ${T_{\rho}\mathcal{O}_{\rho_{0}}\simeq\mathfrak{g}_{\rho}^{\bot}}$
\cite{Rieffel}.

\subsection*{Symplectic structure on $\mathcal{O}_{\rho_{0}}$}

The symplectic form on $\mathcal{O}_{\rho_{0}}$ is given by the Kirillov-Kostant-Souriau
(KKS) form. This is a $2$-form which acts on the tangent space ${T_{\rho}\mathcal{O}_{\rho_{0}}}$
to $\mathcal{O}_{\rho_{0}}$ at any ${\rho\in\mathcal{O}_{\rho_{0}}}$.
Using ${T_{\rho}\mathcal{O}_{\rho_{0}}\simeq\mathfrak{g}_{\rho}^{\bot}}$
it is defined by
\begin{equation}
\omega_{\rho}(\tilde{X}_{\rho},\tilde{Y}_{\rho})=\left(i\rho\,|\,\left[Y,\, X\right]\right)=\left(Y\,|\,\left[X,\, i\rho\right]\right)=\left(X\,|\,\left[i\rho,\, Y\right]\right)\,,\label{eq:KKS}
\end{equation}
where, ${X,Y\in\mathfrak{g}_{\rho}^{\bot}}$ and ${i^{2}=-1}$ ensures
that $i\rho$ is antihermitian and ${\omega_{\rho}(\tilde{X}_{\rho},\tilde{Y}_{\rho})}$
has real value. Clearly when ${X\in\mathfrak{g}_{\rho}}$ or ${Y\in\mathfrak{g}_{\rho}}$
we have ${\omega_{\rho}(\tilde{X}_{\rho},\tilde{Y}_{\rho})=0}$ which
means that indeed $\omega_{\rho}$ is defined on the tangent space
${T_{\rho}\mathcal{O}_{\rho_{0}}}$. One can also check that $\omega$
is closed and non-degenerate. Therefore $\omega$ defines a symplectic
structure on $\mathcal{O}_{\rho_{0}}$. Moreover, group $G$ acts
on $\mathcal{O}_{\rho_{0}}$ in a symplectic way, i.e. ${g_{\ast}\omega=\omega}$,
where $g_{\ast}\omega$ denotes the pullback of $\omega$ by the action
of $g\in G$.

\subsection*{Complex structure on $\mathcal{O}_{\rho_{0}}$}

Having the KKS symplectic form (\ref{eq:KKS}) on $\mathcal{O}_{\rho_{0}}$
there exists a natural almost complex structure associated to it.
It is defined as follows. For ${\rho\in\mathcal{O}_{\rho_{0}}}$ we
compute the polar decomposition of the map ${\mathrm{ad}_{\rho}:\mathfrak{g}\rightarrow\mathfrak{g}}$,
${\mathrm{ad}_{\rho}(X)=\left[i\rho,\, X\right]}$, restricted to ${T_{\rho}\mathcal{O}_{\rho_{0}}\simeq\mathfrak{g}_{\rho}^{\perp}}$.
It is straightforward to see that this restriction is non-degenerate
and that it defines a skew-symmetric operator (with respect to inner
product (\ref{eq:inner})) , ${\mathrm{ad}_{\rho}^{\ast}=-\mathrm{ad}_{\rho}}$.
Therefore the polar decomposition reads

\begin{equation}
\left.\mathrm{ad}_{\rho}\right|_{\mathfrak{g}_{\rho}^{\perp}}=J_{\rho}P_{\rho}\,,\label{eq:adrho}
\end{equation}
where ${P_{\rho}:\mathfrak{g}\rightarrow\mathfrak{g}}$ is a positive
operator, ${J_{\rho}:\mathfrak{g}\rightarrow\mathfrak{g}}$ is orthogonal
and skew-symmetric, ${J_{\rho}^{\ast}=-J_{\rho}}$ and ${\left[J_{\rho},\, P_{\rho}\right]=0}$.
It follows that ${J_{\rho}^{2}=-\mathrm{I}}$. Therefore $J_{\rho}$
can be used to define almost complex structure on $\mathcal{O}_{\rho_{0}}$.
In fact $J_{\rho}$ turns out to be integrable and consequently it
defines the complex structure on $\mathcal{O}_{\rho_{0}}$ \cite{Rieffel}.

\subsection*{Riemannian and K\"{a}hler structures on $\mathcal{O}_{\rho_{0}}$}

The last structure on $\mathcal{O}_{\rho_{0}}$ is the Riemannian
structure that is compatible with $\omega$ and $J$ introduced above.
It is given by the following formula

\begin{equation}
g_{\rho}\left(\tilde{X}_{\rho},\,\tilde{Y}_{\rho}\right)=\omega_{\rho}(\tilde{X}_{\rho},J_{\rho}\tilde{Y}_{\rho})=\left(X\,|\,\left[i\rho,\, J_{\rho}Y\right]\right)\,,\label{eq:metric-sympl}
\end{equation}
for ${\rho\in\mathcal{O}_{\rho_{0}}}$ and ${X,\, Y\in\mathfrak{g}_{\rho}^{\perp}}$.
One easily checks that so defined $g$ is symmetric, positive definite
and $G$-invariant. Moreover, straightforward computation shows that
it is compatible with both $\omega$ and $J$, i.e.

\[
g_{\rho}\left(J_{\rho}\tilde{X}_{\rho},\, J_{\rho}\tilde{Y}_{\rho}\right)=g_{\rho}\left(\tilde{X}_{\rho},\,\tilde{Y}_{\rho}\right),\,\,\, g_{\rho}\left(J_{\rho}\tilde{X}_{\rho},\,\tilde{Y}_{\rho}\right)=\omega_{\rho}(\tilde{X}_{\rho},\,\tilde{Y}_{\rho})\,.
\]
Thus structures $\omega$, $J$ and $g$ define K\"{a}hler structure
on $\mathcal{O}_{\rho_{0}}$. Due to the positive-definiteness of
$g_{\rho}$, $\mathcal{O}_{\rho_{0}}$ is a positive K\"{a}hler manifold
\cite{Symplectic geometry of entanglement}.

\section{Restrictions of geometric structures\label{sec:Restrictions-of-geometric}}

Having defined the relevant geometric structures on $\mathcal{O}_{\rho_{0}}$
we consider the following problem. Let $K$ be a compact semisimple
Lie subgroup of $G$, $K\subset G$. By restriction of the adjoint
action $K$ acts on $\mathcal{O}_{\rho_{0}}$ in a symplectic way.
We denote $K$-orbit through ${\rho\in\mathcal{O}_{\rho_{0}}}$ by $K.\rho$.
Obviously ${K.\rho\subset\mathcal{O}_{\rho_{0}}}$. One can thus consider
the restriction $\omega|_{K.\rho}$ of the symplectic form (\ref{eq:KKS})
to $K.\rho$. The restricted form is still closed, $d\omega|_{K.\rho}=0$,
but it need not to be non-degenerate. As a result, $K$-orbits in
$\mathcal{O}_{\rho}$ need not to be symplectic. We want to know which
of them are. Moreover, we want to know if those which are symplectic
are also K\"{a}hler. Before we state the relevant theorems we review
some background information concerning semisimple Lie algebras. For
more detailed account of this topic consult \cite{Hall}.

\subsection*{Root decomposition of a compact semisimple Lie algebra}

Let $\mathfrak{k}$ be the Lie algebra of $K$. As $K$ is a compact
semisimple Lie group, the algebra $\mathfrak{k}$ has the following
{\it root decomposition} \cite{Hall}

\begin{equation}
\mathfrak{k}=\mathfrak{t}\oplus\bigoplus_{\alpha>0}\mathrm{Span}_{\mathbb{R}}\left(E_{\alpha}-E_{-\alpha}\right)\oplus\bigoplus_{\alpha>0}\mathrm{Span}_{\mathbb{R}}\left(i(E_{\alpha}+E_{-\alpha})\right)\,,\label{eq:root-decomposition}
\end{equation}
where, $\mathfrak{t}$ is a Cartan subalgebra of $\mathfrak{k}$ and
$\alpha$ ranges over all positive roots. The Cartan subalgebra $\mathfrak{t}$
is

\[
\mathfrak{t}=\mathrm{Span}_{\mathbb{R}}\left(iH_{\alpha}:\, H_{\alpha}=[E_{\alpha},\, E_{-\alpha}],\,\alpha>0\right).
\]
Moreover, for each positive root $\alpha$ the triple ${\left\{ E_{\alpha}-E_{-\alpha},\, i(E_{\alpha}+E_{-\alpha}),\, iH_{\alpha}\right\} }$
is isomorphic with $\mathfrak{su}(2)$, i.e.
\begin{eqnarray*}
\left[iH_{\alpha},\, E_{\alpha}-E_{-\alpha}\right] & = & 2i\left(E_{\alpha}+E_{-\alpha}\right)\,,\\
\left[iH_{\alpha},\, i(E_{\alpha}+E_{-\alpha})\right] & = & -2\left(E_{\alpha}-E_{-\alpha}\right)\,,\\
\left[E_{\alpha}-E_{-\alpha},\, i(E_{\alpha}+E_{-\alpha})\right] & = & 2iH_{\alpha}\,.
\end{eqnarray*}

\begin{example}
For ${\mathfrak{k}=\mathfrak{su}(N)}$ the root space decomposition
is particularly simple
\begin{equation}
\mathfrak{su}\left(N\right)=\mathfrak{t}\oplus\bigoplus_{i>j}\mathrm{span}_{\mathbb{R}}\left(X_{ij}\right)\oplus\bigoplus_{i>j}\mathrm{span}_{\mathbb{R}}\left(Y_{ij}\right)\,,\label{eq:root-su(n)}
\end{equation}
where

\[
\mathfrak{t}=\{X\in\mathfrak{k}:\, X-\mbox{diagonal}\}=\mathrm{Span}_{\mathbb{R}}\left(iH_{ij}\right),\, H_{ij}=\left(\kb ii-\kb jj\right)
\]
\begin{equation}
Y_{ij}=\left(\kb ij-\kb ji\right),\,\, X_{ij}=i\left(\kb ij+\kb ji\right),\label{eq:HXY}
\end{equation}
and ${i,j\in\left\{ 1,\ldots,N\right\} }$.
\end{example}

\subsection*{Kostant-Sternberg theorem}

Symplectic orbits ${K.\rho\subset\mathcal{O}_{\rho_{o}}}$ are characterized
by the Kostant-Sternberg theorem \cite{sternberg} (see also \cite{Chruscinski}).
The necessary condition for the orbit ${K.\rho\subset\mathcal{O}_{\rho_{o}}}$
to be symplectic is
\begin{fact}
\label{torus} (The necessary condition) If the orbit ${K.\rho\subset\mathcal{O}_{\rho_{o}}}$
is symplectic (with respect to the restriction of KKS symplectic form
(\ref{eq:KKS})) then there exists ${\tilde{\rho}\in K.\rho}$ such
that ${\left[X,\,\tilde{\rho}\right]=0}$ for all ${X\in\mathfrak{t}}$,
where $\mathfrak{t}$ is a Cartan subalgebra of $\mathfrak{k}$.
\end{fact}
\noindent In the following we assume that the necessary condition
is satisfied, that is, ${\left[\tilde{\rho},\,\mathfrak{t}\right]=0}$.
Using (\ref{eq:tangent-vector}) and the root decomposition (\ref{eq:root-decomposition})
we have

\[
T_{\tilde{\rho}}K.\rho=\bigcup_{\alpha>0}P_{\alpha}\,,
\]
where the sum is over positive roots and
\begin{equation}
P_{\alpha}=\mathrm{Span}_{\mathbb{R}}\left(\left[E_{\alpha}-E_{-\alpha},\,\tilde{\rho}\right],\, i\left[E_{\alpha}+E_{-\alpha},\,\tilde{\rho}\right]\right)\,.\label{eq:P-alpha}
\end{equation}
We will need the following fact whose proof can be found in \cite{Symplectic geometry of entanglement,sternberg}.
\begin{fact}
\label{omega ortogonal}For positive roots ${\alpha\neq\beta}$, ${\omega_{\rho}\left(X,\, Y\right)=0}$
if ${X\in P_{\alpha},}\, {Y\in P_{\beta}}$.
\end{fact}
\noindent Thus ${\left.\omega\right|_{K.\rho}}$ is non-degenerate
if and only if it is non-degenerate on each $P_{\alpha}$ separately.
Using (\ref{eq:KKS}) and (\ref{eq:P-alpha}) it is straightforward
to check:
\begin{fact}
\label{counting-1} Assume that ${\left[\tilde{\rho},\mathfrak{t}\right]=0}$.
Then for any $P_{\alpha}$ we have exactly three possibilities: (1)
${\dim P_{\alpha}=0}$, if and only if ${\left[E_{\alpha},\,\tilde{\rho}\right]=0=\left[E_{-\alpha},\,\tilde{\rho}\right]}$,
(2) ${\dim P_{\alpha}=2}$ and ${\omega|_{P_{\alpha}}=0}$, if and only
if ${\mathrm{tr}\left(\rho H_{\alpha}\right)=0}$ and ${\left[E_{\alpha},\,\tilde{\rho}\right]\neq0}$
or ${\left[E_{-\alpha},\,\tilde{\rho}\right]\neq0}$, (3) ${\dim P_{\alpha}=2}$
and $\omega|_{P_{\alpha}}$ is non-degenerate, if and only if ${\mathrm{tr}\left(\rho H_{\alpha}\right)\neq0}$\textsl{.}
\end{fact}
\noindent We can now state the Kostant-Sternberg theorem in its usual
form.
\begin{fact}
\label{torus-1} (Kostant-Sternberg theorem \cite{sternberg}) The
orbit ${K.\rho\subset\mathcal{O}_{\rho_{o}}}$ is symplectic if and
only if: (1) There exists ${\tilde{\rho}\in K.\rho}$ such that ${\left[X,\,\tilde{\rho}\right]=0}$
for all ${X\in\mathfrak{t}}$ and (2) For any positive root $\alpha$
if ${\mathrm{tr}\left(\rho H_{\alpha}\right)=0}$ then ${\left[E_{\alpha},\,\tilde{\rho}\right]=0=\left[E_{-\alpha},\,\tilde{\rho}\right]}$.
\end{fact}
\noindent  In order to measure how non-symplectic is an orbit $K.\rho$
we will, similarly to \cite{Symplectic geometry of entanglement},
use the notion of degree of degeneracy $\omega|_{K.\rho}$, $D(K.\rho)$.
It is given by

\begin{equation}
D(K.\rho)=\mathrm{dim\,}K.\rho-\mathrm{rank\,}\omega|_{K.\rho}.\label{eq:degenracy}
\end{equation}
The dimension of $K.\rho$ is

\begin{equation}
\mathrm{dim\,}K.\rho=2\left(2\left(\left|\left\{ \alpha|\,\alpha>0\right\} \right|-\left|\left\{ \alpha|\,\alpha>0\,,\left[E_{\alpha},\,\tilde{\rho}\right]=\left[E_{-\alpha},\,\tilde{\rho}\right]=0\right\} \right|\right)\right)\,,\label{eq:dimension-of-orbit}
\end{equation}
where $\left|\mathcal{X}\right|$ denotes the number of elements of
a discrete set $\mathcal{X}$ . The rank of $\omega|_{K.\rho}$ is
\begin{equation}
\mathrm{rank\,}\omega|_{K.\rho}=2\left(\left|\left\{ \alpha|\,\alpha>0\right\} \right|-\left|\left\{ \alpha|\,\alpha>0\,,\mathrm{tr}\left(\tilde{\rho}H_{\alpha}\right)=0\right\} \right|\right).\label{eq:rank}
\end{equation}

\subsection*{Restriction of K\"{a}hler structure }

{
In order to characterize orbits $K.\rho$ that are K\"{a}hler submanifolds of $\mathcal{O}_{\rho_{o}}$
we need the following results
\begin{fact}
\label{possitive Kahler} [\onlinecite{Symplectic geometry of entanglement}] Let
$M$ be a positive K\"{a}hler manifold. Then any complex submanifold
$N\subset M$ is also a K\"{a}hler manifold.
\end{fact}

\begin{prop}\label{complex orbits}
 An orbit $K.{\rho}$ is a complex submanifold of $\mathcal{O}_{\rho_{o}}$ if and only if it is an almost complex sumbanifold of $\mathcal{O}_{\rho_{o}}$.
\begin{proof}
The only thing one has to check is the integrability  is almost-complex of the almost complex structure on $K.{\rho_0}$ that $K.{\rho}$ inherits from  $\mathcal{O}_{\rho_{o}}$. We will deal with this in Appendix 1.
\end{proof}

\end{prop}

Let us now characterize almost complex orbits of $K$ in $\mathcal{O}_{\rho_{o}}$. By the definition of almost complex structure $J_\rho$, an orbit $K.{\rho}$ is almost complex submanifold if and only if for $\tilde{\rho}$ (defined in the Fact 1)
\begin{equation}
\mathrm{ad}_{\tilde{\rho}} \left(T_{\tilde{\rho}} K.\rho \right)\subset T_{\tilde{\rho}} K.\rho  .\label{complex}
\end{equation}
We would like to remark that every simplectic orbit $K.\rho\subset \mathcal{O}_{\rho_{o}}$ can be equipped with the "intrinsic"   K\"{a}hler structure. It follows from the fact\cite{sternberg} that every simplectic orbit $K.\rho$  is diffeomorphic, via the moment map, with coadjoint orbit of some $K$. The latter possess a standard K\"{a}hler structure, as discussed above. We illustrate this phenomenon on a concrete example. Consider the action of the group ${K=SU(2)}$ acting on the projective space ${\mathbb{P}\left(\mathbb{C}^{2j+1}\right)}$ in a natural way (induced from the action of $SU(2)$ on $\mathbb{C}^{2j+1}$ treated as a carrier space of an irreducible representation of $SU(2)$). In this case simplectic orbits of are orbits through states corresponding to nonzero eigenvectors of the operator $\sigma_z$. These orbits are diffeomorphic to two dimensional spheres so they can be equipped with the intrinsic $SU(2)$ invariant K\"{a}hler structure. Nevertheless, only the orbit  through the state corresponding to the maximal or minimal eigenvalue of $\sigma_z$ inherits the  K\"{a}hler structure from ${\mathbb{P}\left(\mathbb{C}^{2j+1}\right)}$. The exhaustive discussion of the above example can be found in [\onlinecite{bengsson}].
}
\section{Geometric description of $CC$ and $CQ$ states\label{sec:geometry of mixed states}}

In the following we apply the ideas presented in sections \ref{sec: coadjoint orbits}
and \ref{sec:Restrictions-of-geometric} to mixed bipartite states.
Let ${\mathcal{H}=\mathcal{H}_{A}\otimes\mathcal{H}_{B}}$, where ${\mathcal{H}_{A}=\mathbb{C}^{N_{1}}}$
and ${\mathcal{H}_{A}=\mathbb{C}^{N_{2}}}$. We start with definitions
of $CC$ and $CQ$ states.

\noindent A quantum state $\rho$ defined on $\mathcal{H}$ is called
a $CC$ state \cite{Horo 1} if it can be written in the form
\begin{equation}
\rho=\sum_{i,j}p_{ij}\kb ii\otimes\kb jj\,,\label{eq:CC-def}
\end{equation}
where $\{\ket i\}_{i=1}^{N_{1}}$ is an orthonormal basis in $\mathcal{H}_{A}$
and $\{\ket i\}_{j=1}^{N_{2}}$ is an orthonormal basis in $\mathcal{H}_{B}$.
 A quantum state $\rho$ defined on $\mathcal{H}$ is called a $CQ$
state if it can be written in the form
\begin{equation}
\rho=\sum_{i}p_{i}\kb ii\otimes\rho_{i}\,,\label{eq:CQ def}
\end{equation}
where $\{\ket i\}_{i=1}^{N_{1}}$ is an orthonormal basis in $\mathcal{H}_{A}$
and $\{\rho_{i}\}_{i=1}^{N_{2}}$ are density matrices defined on
$\mathcal{H}_{B}$.

\noindent In order to use the tools presented in section \ref{sec:Restrictions-of-geometric}
we need to choose some subgroup $K\subset G$. Note that both $CC$
and $CQ$ are ${SU(N_{1})\times SU(N_{2})}$-invariant sets. It turns
out that for $CC$ the relevant group is indeed ${SU(N_{1})\times SU(N_{2})}$.
On the other hand, in order to distinguish geometric properties of
$CQ$ states, one has to take ${SU(N_{1})\times I_{N_{2}}}$.

\subsection{Results for $CC$ states\label{sec:wyniki_geo}}

In the following we prove our main results, i.e. we show that the
orbits through generic $CC$ states of the group ${K=SU(N_{1})\times SU(N_{2})}$,
${K\subset G}$, are the only symplectic orbits in the space of density
matrices on $\mathcal{H}$. We also compute the rank and the dimension
of degeneracy of the symplectic form restricted to $K$-orbits through
$CC$ states. Finally, we illustrate our results by the example two-quibit
system.

The root decomposition (\ref{eq:root-decomposition}) of the Lie algebra
${\mathfrak{k}=\mathfrak{su}(N_{1})\oplus\mathfrak{su}(N_{2})}$ is
a direct sum of root decompositions of $\mathfrak{su}(N_{1})$ and
$\mathfrak{su}(N_{2})$ which are given by (\ref{eq:root-su(n)}).
In the following we use the representation

\begin{equation}
\mathfrak{k}\ni(X_{1},\, X_{2})\mapsto X_{1}\otimes I_{N_{2}}+I_{N_{1}}\otimes X_{2}.\label{eq:rep}
\end{equation}
Under (\ref{eq:rep}) we have

\[
\mathfrak{k}=\mathfrak{t}\oplus\mathfrak{b}^{1}\oplus\mathfrak{b}^{2}\,,
\]
where

\[
\mathfrak{t}=\mathrm{Span}_{\mathbb{R}}\left(iH_{ij}\otimes I_{N_{2}},\, iI_{N_{1}}\otimes H_{kl}:\, i<j,\, k<l\right)\,,
\]

\[
\mathfrak{b}^{1}=\bigoplus_{i>j}\mathrm{Span}_{\mathbb{R}}\left(X_{ij}\otimes I_{N_{2}}\right)\oplus\bigoplus_{i>j}\mathrm{Span}_{\mathbb{R}}\left(Y_{ij}\otimes I_{N_{2}}\right)\,,
\]

\begin{equation}
\mathfrak{b}^{2}=\bigoplus_{k>l}\mathrm{Span}_{\mathbb{R}}\left(I_{N_{1}}\otimes X_{ij}\right)\oplus\bigoplus_{k>l}\mathrm{Span}_{\mathbb{R}}\left(I_{N_{1}}\otimes Y_{ij}\right)\,,\label{eq:root-decomposition-1}
\end{equation}
where matrices $X_{ij}$, $Y_{ij}$ and $H_{ij}$ are defined as in
(\ref{eq:HXY}) and ${i,j\leq N_{1}}$, ${k,l\leq N_{2}}$. We denote
${E_{ij}=\kb ij}$.
\begin{prop}
\label{symplectic-necessary}If the $K$-orbit is symplectic, then
it consists only of $CC$ states, i.e. $CC$ states satisfy the necessary
condition given in Fact \ref{torus}.\end{prop}
\begin{proof}
Assume $K.\rho$ is symplectic. By Fact \ref{torus}, ${\left[\mathfrak{t},\,\rho\right]=0}$.
But $\mathfrak{t}$ consists of matrices ${X_{1}\otimes I_{N_{2}}+I_{N_{1}}\otimes X_{2}}$,
where $X_1,X_2$ are traceless diagonal matrices. It is thus clear that ${\rho=\sum_{ik}\,p_{ik}E_{ii}\otimes E_{kk}}$.
Conversely, if ${\rho=\sum_{ik}\,p_{ik}E_{ii}\otimes E_{kk}}$ then
it follows directly from the definition of $\mathfrak{t}$ that ${\left[\mathfrak{t},\,\rho\right]=0}$.\end{proof}
\begin{prop}
\label{sympl cond}The orbit of adjoint action of ${K=SU(N_{1})\times SU(N_{2})}$
through CC state ${\rho=\sum_{i,j}\,p_{ij}\kb ii\otimes\kb jj}$ is symplectic
if and only if the following conditions hold: \end{prop}
\begin{enumerate}
\item for any ${i\neq j}$, ${\sum_{l}\,p_{il}=\sum_{l}\,p_{jl}}\Longrightarrow{\forall_l\, p_{il}=p_{jl}}$,
\item for any ${i\neq j}$, ${\sum_{l}\,p_{li}=\sum_{l}\,p_{lj}}\Longrightarrow{\forall_{l}\,p_{li}=p_{lj}}$. \end{enumerate}
\begin{proof}
Let ${\rho=\sum_{i,j}\,p_{ij}E_{ii}\otimes E_{jj}}$. By Fact \ref{torus-1}
we first verify when ${\mathrm{tr}(\rho H_{\alpha})=0}$ :
\begin{eqnarray}
\mathrm{tr}\left(\rho\left(H_{ij}\otimes I_{N_{2}}\right)\right) & = & \mathrm{tr}\left(\sum_{l=1}^{N_{2}}\left(p_{il}E_{ii}-p_{jl}E_{jj}\right)\otimes E_{ll}\right)\\
 & = & \sum_{l=1}^{N_{2}}p_{il}-\sum_{l=1}^{N_{2}}p_{jl}\,,\\
\mathrm{tr}\left(\rho\left(H_{ij}\otimes I_{N_{2}}\right)\right) & = & 0\Leftrightarrow\sum_{l=1}^{N_{2}}p_{il}=\sum_{l=1}^{N_{2}}p_{jl}\,.\label{eq:rank1}
\end{eqnarray}
Similarly
\begin{equation}
\mathrm{tr}\left(\rho\left(I_{N_{1}}\otimes H_{ij}\right)\right)=0\Leftrightarrow\sum_{l=1}^{N_{1}}p_{li}=\sum_{l=1}^{N_{1}}p_{lj}\,.\label{eq:rank2}
\end{equation}
Next, we verify when ${[E_{\alpha},\rho]=[E_{-\alpha},\rho]=0}$.
\begin{eqnarray}
\left[E_{ij}\otimes I_{N_{2}},\rho\right] & = & \sum_{l=1}^{N_{2}}(p_{jl}E_{ij}-p_{il}E_{ij})\otimes E_{ll}\,,\\
\left[E_{ji}\otimes I_{N_{2}},\rho\right] & = & \sum_{l=1}^{N_{2}}(p_{il}E_{ji}-p_{jl}E_{ji})\otimes E_{ll}\,,\\
\left[E_{ij}\otimes I_{N_{2}},\rho\right] & = & \left[E_{ji}\otimes I_{N_{1}},\rho\right]=0\Leftrightarrow\forall_{l}\, p_{jl}=p_{il}\,.\label{eq:p-alpha-0}
\end{eqnarray}
Analogously
\begin{equation}
\left[I_{N_{1}}\otimes E_{ij},\rho\right]=\left[I_{N_{1}}\otimes E_{ji},\rho\right]=0\Leftrightarrow\forall\, p_{lj}=p_{li}\,.\label{eq:p-alpha-01}
\end{equation}
Therefore the condition ${\mathrm{tr}\left(\rho H_{\alpha}\right)=0}\Rightarrow{\left[E_{\alpha},\rho\right]=0}$
translates to the conditions 1 and 2 above.
\end{proof}
One can interpret results stated in the Proposition \ref{sympl cond}
in terms of the reduced density matrices of $\rho$. First, the $K$-orbit
through a $CC$ state $\rho$ is symplectic, if spectra of ${\rho_{1}=\mathrm{tr}_{2}\left(\rho\right)}$
and ${\rho_{2}=\mathrm{tr}_{1}\left(\rho\right)}$ are non-degenerate.
Moreover, whenever there is a pair of equal eigenvalues in the spectrum
of $\rho_{1}$ or $\rho_{2}$, the $K$-orbit through $\rho$ is symplectic
provided $\rho$ satisfies conditions 1 and 2 stated in Proposition
\ref{sympl cond}.
\begin{cor}
For generic $CC$ state the spectra of $\rho_{1}$ and $\rho_{2}$
are non-degenerate. Therefore, the set of all $CC$ states is the
closure of all symplectic $K$-orbits in the space of quantum states.
\label{cc-closure}\end{cor}
\begin{proof}
By Proposition \ref{sympl cond} the degeneracies in the spectra of
$\rho_{1}$ and $\rho_{2}$ are described by equations for hyperplanes
in the set of $CC$ states for fixed bases $\left\{ \ket i\right\} $
and $\left\{ \ket j\right\} $ in $\mathcal{H}_{A}$ and $\mathcal{H}_{B}$
respectively, i.e.
\begin{equation}
\sum_{l}p_{il}=\sum_{l}p_{jl},\,\,\sum_{l}p_{li}=\sum_{l}p_{lj}
\end{equation}
for some $i\neq j$. There is a finite number of them, and thus the
complement of the set described by them is dense in the set of all
CC states.
\end{proof}

{
Let us now use \eqref{complex} to study which symplectic orbits of ${SU(N_1)\times SU(N_2)}$ are actually K\"{a}hler (they inherit the K\"{a}hler structure form $\mathcal{O}_{\rho_{o}}$) . Straightforward computations based on \eqref{eq:p-alpha-0} and \eqref{eq:p-alpha-01} show that the tangent space to the orbit of ${SU(N_1)\times SU(N_2)}$ at the state $\rho$ is spanned by the following vectors (${i>j}$).

\begin{eqnarray}
\left[X_{ij}\otimes I_{N_{2}},i\rho\right] & = & \sum_{l=1}^{N_{2}}(p_{il}-p_{jl})Y_{ij}\otimes E_{ll}\,,\\
\left[Y_{ij}\otimes I_{N_{2}},i\rho\right] & = & \sum_{l=1}^{N_{2}}   (p_{jl}-p_{il}) i X_{ij}\otimes E_{ll}\,,\\
\left[I_{N_1}\otimes X_{ij},i\rho\right] & = & \sum_{k=1}^{N_{1}}  (p_{ki}-p_{kj}) E_{kk}\otimes  Y_{ij}\,,\\
\left[I_{N_1}\otimes Y_{ij},i\rho\right] & = & \sum_{k=1}^{N_{1}}  (p_{kj}-p_{ki})  E_{kk}\otimes  i X_{ij}\ .\ \label{eq:tangent space}
\end{eqnarray}

The action of $\mathrm{ad}_\rho$ on each ${P_\alpha\subset T_\rho K.\rho}$ (see \eqref{eq:root-decomposition-1} for the convention used to describe roots of ${\mathfrak{su}(N_1)\oplus \mathfrak{su}(N_2)}$)
\begin{eqnarray}
\left[\left[X_{ij}\otimes I_{N_{2}},i\rho\right],i\rho\right] & = & - \sum_{l=1}^{N_{2}}(p_{jl}-p_{il})^2 i X_{ij}\otimes E_{ll}\ ,\\ 
\left[\left[Y_{ij}\otimes I_{N_{2}},i\rho\right],i\rho\right] & = & \sum_{l=1}^{N_{2}}(p_{il}-p_{jl})^2 Y_{ij}\otimes E_{ll}\ ,\\
\left[\left[I_{N_1}\otimes X_{ij},i\rho\right] ,i\rho\right] & = &  - \sum_{l=1}^{N_{1}}(p_{kj}-p_{ki})^2 i E_{kk} \otimes X_{ij} \ ,\\  
\left[\left[I_{N_1}\otimes Y_{ij},i\rho\right],i\rho\right] & = &  \sum_{k=1}^{N_{1}}  (p_{ki}-p_{kj})^2   E_{kk} \otimes Y_{ij}\ .  \label{eq:complex local}
\end{eqnarray}

We have ${\mathrm{ad}_\rho \left(P_\alpha\right)\cap P_\beta ={0}}$ for ${\alpha\neq \beta}$. Direct inspection shows that in order for the orbit through the state $\rho$ to be K\"{a}hler the following conditions have to be satisfied,
\begin{eqnarray}
\left(p_{il}-p_{jl}\right)\left(\left(p_{il}-p_{jl}\right)- \beta_{i,j} \right)=0 \ ,   \ \forall (i,j)\ ,\ N_1\geq i>j \geq 1 \   \forall\ l=1,\ldots,N_2 \ , \label{condcomplex1} \\
\left(p_{ki}-p_{kj}\right)\left(\left(p_{ki}-p_{kj}\right)- \gamma_{i,j} \right)=0  \ , \forall (i,j)\ ,\ N_2\geq i>j \geq 1 \   \forall\ k=1,\ldots,N_1 \ , \label{condcomplex2}
\end{eqnarray}
where $\beta_{i,j}$ and $\gamma_{i,j}$ are real paremetres depending only on indices $i$ and $j$. 
\begin{prop} {\label{kahlerlu} }The solutions to conditions \eqref{condcomplex1} and \eqref{condcomplex2} are the following:

\end{prop}
\begin{enumerate}
	\item ${p_{i_0 j_0}=1}$ and ${p_{ij}=0}$ for ${(i,j)\neq (i_0,j_0)}$. 
	\item ${p_{ij}=\frac{1}{N_1 N_2}}$ for all pairs of indices $(i,j)$. 
\end{enumerate}

\begin{proof}
One easily checks that the above satisfy  \eqref{condcomplex1} and \eqref{condcomplex2}. Assuming that there exist two pairs of indices $(i,j)$ and $(i',j')$ such that ${0\neq p_{ij}\neq p_{i'j'} \neq 0}$ leads to a contradiction.
\end{proof}

The above reasoning reproduces results from [\onlinecite{Chruscinski}] where the author showed that pure separable states form the unique  K\"{a}hler orbit of the complexification of $K$, $K^{\mathbb{C}}$, in $\mathcal{O}_{\rho_0}$. However, we would like to point out that in general there might be more  K\"{a}hler orbits of ${K=SU(N_{1})\times SU(N_{2})}$ in the set of density matrices. Proposition \ref{kahlerlu} treats this problem.}

We now turn to a detailed description of symplectic properties of
orbits through $CC$ states. We compute dimensions of orbits, ${\mathrm{dim}\, K.\rho}$,
rank of the symplectic form restricted to orbits ${\mathrm{rk}\left.\omega\right|_{K.\rho}}$
, and its degree of degeneracy, ${D\left(K.\rho\right)}$. We first
introduce the notation which will be used in formulas for these quantities.

For fixed ${\rho\in\mathcal{CC}}$ we consider the coefficients $p_{ij}$
of $\rho$ as entries of the ${N_{1}\times N_{2}}$ matrix $P$. Let
$R_{i}$ be its $i$-th row and $C_{j}$ be its $j$-th column. Let
$S(X)$ denote the sum of elements of $X$, for $X$ being either
a row or a column. Define
\begin{equation}
\mathrm{\mathcal{SR}}=\{S(R_{i})\}_{i=1}^{i=N_{1}},\,\mathrm{\mathcal{SC}}=\{S(C_{j})\}_{j=1}^{j=N_{2}}\,.
\end{equation}
That is, $\mathcal{SR}$ and $\mathcal{SC}$ consist of all numbers
one can get by summing elements in rows and columns of $P$ respectively.
For each ${r\in\mathrm{\mathcal{SR}}}$ and ${c\in\mathcal{SC}}$, let
\begin{equation}
\mathcal{I}_{r}=\{i:S(R_{i})=r\},\,\mathcal{J}_{c}=\{j:\, S(C_{j})=c\}\,,
\end{equation}
be sets consisting of indices that label rows and columns of $P$
whose sums of elements are equal to $r$ and $c$ respectively. Of
course, ${\{1,\ldots,N_{1}\}=\bigcup_{r\in\mathcal{SR}}\mathcal{I}_{r}}$
and ${\{1,\ldots,N_{2}\}=\bigcup_{c\in\mathcal{SC}}\mathcal{I}_{c}}$,
where $\bigcup$ denotes the union of sets. Moreover, for each ${r\in\mathrm{\mathcal{SR}}}$
and ${c\in\mathcal{SC}}$, let

\begin{equation}
\mathrm{\mathcal{R}}_{r}=\{R_{i}:S(R_{i})=r\},\,\mathrm{\mathcal{C}}_{c}=\{C_{j}:S(C_{j})=c\}\,,
\end{equation}
be the sets consisting of rows and columns of $P$ whose sums equal
$r$ and $c$ respectively. Let
\begin{equation}
\mathcal{R}=\bigcup_{r\in\mathcal{SR}}\mathcal{R}_{r},\,\mathcal{C}=\bigcup_{c\in\mathrm{\mathcal{SC}}}\mathrm{\mathcal{C}}_{c}\,,
\end{equation}
be the sets consisting of all rows and columns of $P$. For each ${R\in\mathrm{\mathcal{R}}},{C\in\mathcal{C}}$
let
\begin{equation}
\mathcal{I}_{R}=\{i:R_{i}=R\},\,\mathcal{J}_{C}=\{j:C_{j}=C\}\,,
\end{equation}
be the sets consisting of indices that label rows and columns of $P$
that equal $R$ and $C$ respectively. Of course for each row $R$
we have ${\mathcal{I}{}_{R}\subset\mathcal{I}_{S(R)}}$ and for each
column $C$ we have ${\mathcal{I}_{C}\subset\mathcal{I}_{S(C)}}$. Moreover
for each ${r\in\mathrm{\mathcal{SR}}}$, ${\mathcal{I}{}_{r}=\bigcup_{R\in\mathrm{\mathcal{R}}_{r}}\mathcal{I}_{R}}$
and for each ${c\in\mathrm{\mathcal{SC}}}$, ${\mathcal{J}_{c}=\bigcup_{C\in\mathrm{\mathcal{\mathcal{C}}}}\mathcal{J}_{C}}$.
Finally, let us denote by $\left|\mathcal{X}\right|$ the number of
elements of a discrete set $\mathcal{X}$ and by $\binom{a}{b}$ the
binomial coefficient. In what follows we assume the convention ${\binom{a}{b}=0}$
for $a<b$.

\noindent  By (\ref{eq:dimension-of-orbit}) to calculate ${\mathrm{dim}\, K.\rho}$
it is enough to determine when ${\left[E_{\alpha},\,\rho\right]=0=\left[E_{-\alpha},\,\rho\right]}$.
Using (\ref{eq:p-alpha-0}) and (\ref{eq:p-alpha-01}) we have

\begin{eqnarray}
\left[E_{ij}\otimes I_{N_{2}},\,\rho\right] & = & 0\Longleftrightarrow i,j\in\mathcal{I}_{R}\,\,\mathrm{for}\,\mathrm{some}\, R\in\mathcal{R}\,,\\
\left[I_{N_{1}}\otimes E_{ij},\,\rho\right] & = & 0\Longleftrightarrow i,j\in\mathcal{J}_{C}\,\,\mathrm{for}\,\mathrm{some}\, C\in\mathcal{C}.
\end{eqnarray}
Hence

\begin{equation}
\mathrm{dim}\, K.\rho=2\left(\binom{N_{1}}{2}-\sum_{R\in\mathcal{R}}\binom{\left|\mathcal{I}_{R}\right|}{2}+\binom{N_{2}}{2}-\sum_{C\in\mathcal{C}}\binom{\left|\mathcal{J}_{S}\right|}{2}\right)\,.\label{eq:dimK-CC}
\end{equation}

\noindent By (\ref{eq:rank}) to calculate ${\mathrm{rank}\,\left.\omega\right|_{K.\rho}}$
it is enough to determine when ${\mathrm{tr}\left(\rho H_{\alpha}\right)=0}$.
By (\ref{eq:rank1}) and (\ref{eq:rank2}) we have
\begin{eqnarray}
\mathrm{tr}\left(\rho\left(H_{ij}\otimes I_{N_{2}}\right)\right) & = & 0\Longleftrightarrow i,j\in\mathcal{I}_{r}\,\,\mathrm{for}\,\mathrm{some}\, r\in\mathcal{SR}\,,\\
\mathrm{tr}\left(\rho\left(I_{N_{1}}\otimes H_{ij}\right)\right) & = & 0\Longleftrightarrow i,j\in\mathcal{J}_{c}\,\,\mathrm{for}\,\mathrm{some}\, c\in\mathcal{SC}\,.
\end{eqnarray}
Thus

\begin{equation}
\mathrm{rank}\,\left.\omega\right|_{K.\rho}=2\left(\binom{N_{1}}{2}-\sum_{r\in\mathcal{SR}}\binom{\left|\mathcal{I}_{r}\right|}{2}+\binom{N_{2}}{2}-\sum_{c\in\mathcal{SC}}\binom{\left|\mathcal{J}_{c}\right|}{2}\right),\label{eq:Rank-CC}
\end{equation}
Having established formulas for ${\mathrm{dim}\, K.\rho}$ and ${\mathrm{rank}\,\left.\omega\right|_{K.\rho}}$
we arrive at our final result:
\begin{prop}
\label{mainCC} The dimension of the degeneracy of the symplectic
form on the orbit through the $CC$ state ${\rho=\sum_{i,j}\,p_{ij}E_{ii}\otimes E_{jj}}$
is equal to
\begin{equation}
\mathrm{D}(K.\rho)=2\left(\sum_{r\in\mathrm{\mathcal{SR}}}\binom{\left|\mathcal{I}_{r}\right|}{2}-\sum_{R\in\mathcal{R}}\binom{\left|\mathcal{I}_{R}\right|}{2}+\sum_{c\in\mathrm{\mathcal{SC}}}\binom{\left|\mathcal{J}_{c}\right|}{2}-\sum_{C\in\mathcal{C}}\binom{\left|\mathcal{J}_{C}\right|}{2}\right).\,\label{eq:gegeneracy CC}
\end{equation}

\end{prop}
$CC$ states, for which the corresponding $K$-orbits have the maximal
degree of degeneracy ${\mathrm{D}(K.\rho)}$, correspond to so-called
magic rectangles \cite{rectangles}. To be more precise, the maximum
in equation (\ref{eq:gegeneracy CC}) is obtained when ${\mathrm{\left|\mathcal{R}\right|}=N_{1}}$,
${\mathrm{\left|\mathcal{C}\right|}=N_{2}}$ and both $\mathcal{SR}$
and $\mathcal{SC}$ have precisely one element. Translating these
conditions to the properties of ${N_{1}\times N_{2}}$ matrix $p_{ij}$
encoding a given $CC$ state $\rho$, one arrives at the following
conditions:
\begin{enumerate}
\item Each row and colum of $p_{ij}$ have to consist of different elements.
\item Sums of elemenst in each row are the same. The same concerns sums
of elements in each column.
\end{enumerate}

\subsection*{Two-qubit $CC$ states}

Let ${\rho=p_{11}E_{11}\otimes E_{11}+p_{12}E_{11}\otimes E_{22}+p_{21}E_{22}\otimes E_{11}+p_{22}E_{22}\otimes E_{22}}$.
We will now use Proposition \ref{mainCC} to calculate dimensions
of orbits, ranks of the form $\omega_{\rho}$, and its degrees of
degeneration. Therefore, let us consider the matrix
\[
P=\left(\begin{array}{cc}
p_{11} & p_{12}\\
p_{21} & p_{22}
\end{array}\right)\,.
\]
Of course, ${1=\mathrm{Tr}\rho=p_{11}+p_{12}+p_{21}+p_{22}}$. There
are four possibilities:

\paragraph{Distinct sums both in columns and in rows}

Consider the case ${p_{11}+p_{12}\neq p_{21}+p_{22}}$, ${p_{11}+p_{21}\neq p_{12}+p_{22}}$.
It follows that no two columns or two rows are identical, therefore
due to Proposition \ref{mainCC}, ${\mathrm{dim}\, K.\rho=2\left(\binom{2}{2}+\binom{2}{2}\right)=4}$.
The rank of $\left.\omega\right|_{K.\rho}$ is also $4$, so it is
non-degenerate and the orbit is symplectic.

\paragraph{Equal sums in columns, distinct sums in rows}

Consider the case ${p_{11}+p_{12}\neq p_{21}+p_{22}}$, ${p_{11}+p_{21}=p_{12}+p_{22}}$.
Because ${\mathrm{Tr}\rho=1}$, it follows that ${p_{11}+p_{21}=p_{21}+p_{22}=\frac{1}{2}}$,
and setting ${\alpha=p_{11}},{\beta=p_{22}}$ gives us ${p_{21}=\frac{1}{2}-\alpha}$,
${p_{12}=\frac{1}{2}-\beta}$ and ${\alpha\neq\beta}$, that is ${P=\left(\begin{array}{cc}
\alpha & \frac{1}{2}-\beta\\
\frac{1}{2}-\alpha & \beta
\end{array}\right).}$ Due to Proposition \ref{mainCC}, the rank of $\omega_{\rho}$ is
equal to ${2\left(\binom{2}{2}+\binom{2}{2}-\binom{2}{2}\right)=2}$,
because there is a pair of columns with equal sums. Now, if these
columns are equal, then ${\alpha+\beta=\frac{1}{2}}$, meaning ${P=\left(\begin{array}{cc}
\alpha & \alpha\\
\frac{1}{2}-\alpha & \frac{1}{2}-\alpha
\end{array}\right)}$. Condition ${\alpha\neq\beta}$ implies ${\alpha\neq\frac{1}{4}}$. From
Proposition \ref{mainCC} it follows that ${\mathrm{dim}\, K.\rho=2\left(\binom{2}{2}+\binom{2}{2}-\binom{2}{2}\right)=2}$,
as there are two identical columns and no identical rows. The orbit
is symplectic of dimension 2. For any ${\beta\neq\frac{1}{2}-\alpha}$
and ${\beta\neq\alpha}$, we get ${\mathrm{dim}\, K.\rho=4}$,
so the degeneracy of ${\left.\omega\right|_{K.\rho}}$ is equal to ${\mathrm{D}(K.\rho)=2}$.

\paragraph{Distinct sums in columns, equal sums in rows}

Similarly to the case above, we consider $\rho$ such that ${P=\left(\begin{array}{cc}
\alpha & \frac{1}{2}-\alpha\\
\frac{1}{2}-\beta & \beta
\end{array}\right)}$ and ${\beta\neq\alpha}$. We get that if ${\alpha+\beta=\frac{1}{2}}$
and ${\alpha\neq\frac{1}{4}}$, then the orbit is symplectic of dimension
2, and if ${\beta\neq\frac{1}{2}-\alpha}$ and ${\beta\neq\alpha}$, then
the orbit is of dimension $4$ and $\left.\omega\right|_{K.\rho}$
has degeneracy $2$.

\paragraph{Equal sums in columns and rows}

What remains is the case ${p_{11}+p_{12}=p_{21}+p_{22}}$, ${p_{11}+p_{21}=p_{12}+p_{22}}$.
With ${\alpha=p_{11}}$, simple calculations lead us to ${P=\left(\begin{array}{cc}
\alpha & \frac{1}{2}-\alpha\\
\frac{1}{2}-\alpha & \alpha
\end{array}\right)}$. Now, if ${\alpha=\frac{1}{4}}$, then both columns and both rows are
equal, and the orbit is of dimension $0$. In fact, it is just a point
${\rho=\frac{1}{4}\cdot Id}$. For any ${\alpha\neq\frac{1}{4}}$, both
of the columns are distinct, and so are the rows. Therefore, due to
Proposition \ref{mainCC}, ${\mathrm{dim}\, K.\rho=4}$ and ${\mathrm{rk}\,\left.\omega\right|_{K.\rho}=0}$,
that is the degeneracy of $\left.\omega\right|_{K.\rho}$ is equal
to ${\mathrm{D}(K.\rho)=4}$, and is maximal possible for a $CC$ state.

Because in the case considered $CC$ states in a fixed computational
basis ${\ket 1\ket 1}$, ${\ket 1\ket 2}$, ${\ket 2\ket 1}$ and ${\ket 2\ket 2}$
form a three-dimensional simplex, it is possible to draw pictures
illustrating the results discussed above. We denote vertices of the
simplex we denote by ${E_{11}\otimes E_{11}}$, ${E_{11}\otimes E_{22}}$
, ${E_{22}\otimes E_{11}}$ and ${E_{22}\otimes E_{22}}$. Figures \ref{fig1}-\ref{fig3}
show dimension of $LU$-orbits as well as the rank and degree
of degeneracy for different points in this simplex.

\begin{figure}[H]
\begin{centering}
\includegraphics[width=6cm]{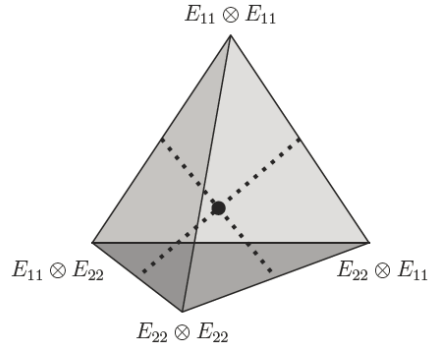}
\par\end{centering}

\caption{\label{fig1} Dimensions of orbits through $CC$ states of two qbits. Large dot:
${\dim\, K.\rho=0}$, dotted lines: ${\dim\, K.\rho=2}$, elsewhere: ${\dim\, K.\rho=4}.$}
\end{figure}

\begin{figure}[H]
\begin{centering}
\includegraphics[width=6cm]{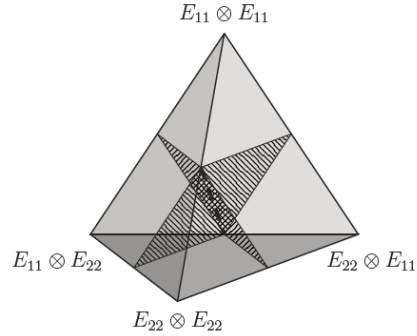}
\par\end{centering}

\caption{ \label{fig2} Ranks of $\left.\omega\right|_{K.\rho}$ for orbits through $CC$
states of two qbits. Thick dashed line: ${\mathrm{rk}\left.\omega\right|_{K.\rho}=0}$
, lined surfaces: ${\mathrm{rk}\left.\omega\right|_{K.\rho}=2}$, elsewhere:
${\mathrm{rk}\left.\omega\right|_{K.\rho}=4}$.}
\end{figure}

\begin{figure}[H]
\begin{centering}

\includegraphics[width=6cm]{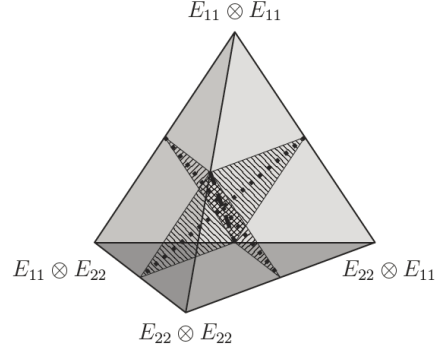}
\par
\end{centering}

\caption{ \label{fig3} Degrees of degeneracy of $\left.\omega\right|_{K.\rho}$ for orbits
through $CC$ states of two qbits. Thick dashed line: ${\mathrm{D}(K.\rho)=4}$,
lined surfaces: ${\mathrm{D}(K.\rho)=2}$, dotted lines and elsewhere:
${\mathrm{D}(K.\rho)=0}$.}
\end{figure}

\subsection{Results for CQ states}

Setting ${K=SU(N_{1})\times I{}_{N_{2}}}$ give us analogous results
for $CQ$ states as we got for $CC$ states. In this case the Lie
algebra of $K$, ${\mathfrak{k}=\mathfrak{su}(N_{1})}$ is represented
on $\mathcal{H}$ {\it via} the mapping

\begin{equation}
\mathfrak{k}\ni X\mapsto X\otimes I_{N_{2}}.\label{eq:rep2}
\end{equation}
Throughout the computations we will use the notation analogous to
the one used for $CC$ states:

\begin{equation}
\mathfrak{k}=\mathfrak{t}\oplus\bigoplus_{i>j}\mathrm{Span}_{\mathbb{R}}\left(X_{ij}\otimes I_{N_{2}}\right)\oplus\bigoplus_{i>j}\mathrm{Span}_{\mathbb{R}}\left(Y_{ij}\otimes I_{N_{2}}\right),\label{eq:root-decomposition2}
\end{equation}

\[
\mathfrak{t}=\mathrm{Span}_{\mathbb{R}}\left(iH_{ij}\otimes I_{N_{2}},\,\, i<j\right)\,,
\]
where matrices $X_{ij}$, $Y_{ij}$ and $H_{ij}$ are defined as in
(\ref{eq:HXY}) and $i,\, j\leq N_{1}$.
\begin{prop}
If the orbit of adjoint action of $K=SU(N_{1})\times I_{N_{2}}$  is
symplectic, then it consists only of CQ states.\label{torus Cq}\end{prop}
\begin{proof}
Let $K.\rho$ be symplectic. By Fact \ref{torus}, we may assume that
${\left[\mathfrak{t},\,\rho\right]=0}$. Thus $\rho$ and elements from
$\mathfrak{t}$ must have common eigenvalues. Therefore, $\rho$ has
the form of the CQ state. \end{proof}
\begin{prop}
The orbit of the coadjoint action of ${K=SU(N_{1})\times I_{N_{2}}}$
through a $CQ$ state ${\rho=\sum_{i}\,p_{i}\kb ii\otimes\rho_{i}}$ is
symplectic if and only if for any ${i\neq j}$, ${p_{i}=p_{j}}\Longrightarrow {p_{i}\rho_{i}=p_{j}\rho_{j}}$.\label{sympl cq}\end{prop}
\begin{proof}
Let ${\rho=\sum_{i}\,p_{i}E_{ii}\otimes\rho_{i}}$. Following Fact \ref{torus},
we check for which $\alpha$ we have ${\mathrm{tr}(\rho H_{\alpha})=0}$.
\begin{equation}
\mathrm{tr}(\rho(H_{ij}\otimes I_{N_{2}}))=\mathrm{tr}(p_{i}E_{ii}\otimes\rho_{i}-p_{j}E_{jj}\otimes\rho_{j})=p_{i}-p_{j}
\end{equation}
\begin{equation}
\mathrm{tr}(\rho(H_{ij}\otimes I_{N_{2}}))=0\iff p_{i}=p_{j}\label{eq:degeneracy cq}
\end{equation}
We also check the condition ${[E_{\alpha},\rho]=[E_{-\alpha},\rho]=0}$.

\begin{equation}
[E_{ij}\otimes I_{N_{2}},\rho]=p_{j}E_{ij}\otimes\rho_{j}-p_{i}E_{ij}\otimes\rho_{i}\label{orbit-dim-1}
\end{equation}
\[
[E_{ji}\otimes I_{N_{2}},\rho]=p_{i}E_{ji}\otimes\rho_{i}-p_{j}E_{ji}\otimes\rho_{j}
\]
\[
[E_{ij}\otimes I_{N_{2}},\rho]=[E_{ji}\otimes I_{N_{2}},\rho]=0\iff p_{j}\rho_{j}=p_{i}\rho_{i}
\]
So ${\mathrm{tr}(\rho H_{\alpha})=0}\Longrightarrow{[E_{\alpha},\rho]=0}$
for all positive roots $H_{\alpha}$ translate to the conditions given
above.

\end{proof}
\begin{cor}
A generic orbit of the adjoint action of ${K=SU(N_{1})\times I_{N_{2}}}$
through a $CQ$ state is symplectic and the set of all CQ states is
equal to the closure of the sum of all symplectic orbits of ${SU(N_{1})\times I_{N_{2}}}$
in the space of quantum states. \end{cor}
\begin{proof}
Analogous to the proof of Corollary \ref{cc-closure}.
\end{proof}

We now give more detailed description of $K$-orbits through $CQ$
states. We compute dimensions of orbits, ${\mathrm{dim}\, K.\rho}$
, rank of the symplectic form restricted to orbits ${\mathrm{rk}\left.\omega\right|_{K.\rho}}$,
and its degree of degeneracy, ${D\left(K.\rho\right)}$. We first introduce
necessary notation which we then use to compute numbers we are interested
in. Let us fix ${\rho\in\mathcal{CQ}}$ and let
\begin{equation}
\mathcal{P}=\{p_{i}\}_{i=1}^{N_{1}}
\end{equation}
 be the set of probabilities that appear in (\ref{eq:CQ def}) . For
each ${p\in\mathcal{P}}$, let
\begin{equation}
\mathcal{I}_{p}=\{i:p_{i}=p\}\,,
\end{equation}
 be the set of indices that have the same value of $p$. Of course
${\{1,\ldots,N_{1}\}=\bigcup_{p\in\mathcal{P}}\mathcal{I}_{p}}$. For
${p\in\mathcal{P}}$, let
\begin{equation}
\mathrm{\mathcal{Q}}_{p}=\{p_{i}\rho_{i}:p_{i}=p\}\,.
\end{equation}
Let
\begin{equation}
\mathcal{Q}=\{p_{i}\rho_{i}\}=\bigcup_{p\in\mathcal{P}}Q_{p}\,,\label{eq:def Q}
\end{equation}
be the set consisting of all ${p_{i}\rho_{i}}$ that appear in the decomposition
(\ref{eq:CQ def}). For each ${\sigma\in\mathcal{Q}}$, let
\begin{equation}
\mathcal{I}_{\sigma}=\{i:p_{i}\rho_{i}=\sigma\}\,.\label{eq:def I sigma}
\end{equation}
Of course for each ${\sigma\in\mathcal{Q}}$ we have ${\mathcal{I}_{\sigma}\subset\mathcal{I}_{\mathrm{tr}\sigma}}$.
Moreover for each ${p\in\mathcal{P}}$, ${\mathcal{I}_{p}=\bigcup_{\sigma\in\mathcal{Q}_{p}}\mathcal{I}_{\sigma}}$.
 As before, we assume the convention ${\binom{a}{b}=0}$ for ${a<b}$.

We essentially repeat arguments that were given to justify Proposition
\ref{mainCC}. The only difference now is the structure of the group
$K$. One should introduce minor corrections in the argumentation.
In particular, sets $\mathcal{R}$ and $\mathcal{C}$ should be replaced
by the set $\mathcal{Q}$. Also, the set $\mathcal{P}$ should be
taken instead of the sets $\mathcal{SR}$ and $\mathcal{SC}$. By
the analogy we obtain the following formulas:

\begin{equation}
\mathrm{dim\,}K.\rho=2\left(\binom{N_{1}}{2}-\sum_{\sigma\in\mathcal{Q}}\binom{\left|\mathcal{I}_{\sigma}\right|}{2}\right),
\end{equation}

\begin{equation}
\mathrm{rk}\left.\omega\right|_{K.\rho}=2\left(\binom{N_{1}}{2}-\sum_{p\in\mathcal{P}}\binom{\left|\mathcal{I}_{p}\right|}{2}\right).
\end{equation}
We can now state the final result.
\begin{prop}
\label{main CQ}The dimension of the degeneracy on the $K$-orbit
through a $CQ$ state $\rho$ is equal to
\begin{equation}
\mathrm{D}(K.\rho)=2\left(\sum_{p\in\mathcal{P}}\binom{\left|\mathcal{I}_{p}\right|}{2}-\sum_{\sigma\in\mathcal{Q}}\binom{\left|\mathcal{I}_{\sigma}\right|}{2}\right).
\end{equation}

\end{prop}

\section{Euler-Poincar\'{e} characteristics of homogenous spaces\label{sec:topology}}

Having discussed the geometric properties of $CC$ and $CQ$ states
we focus now on their topological features. In the following we recall
the notion of the Euler-Poincar\'{e} characteristics for compact
manifolds and homogenous spaces. In particular we invoke the Hopf-Samelson
theorem which will be used in the Section \ref{sec:Topol_results} for caclulation
of Euler-Poincar\'{e} characteristics of $K$-orbits through $CC$
and $CQ$ states.

Let $K$ be a compact connected Lie group. Assume that $K$ acts on
a compact manifold $M$. It is known that each orbit ${\mathcal{O}_{x}=\{g.x:\, g\in K\}}$
of $K$-action on $M$ is a homogenous space, i.e. ${\mathcal{O}_{x}=K/K_{x}}$,
where $K_{x}$ is the isotropy subgroup of $x$, that is, ${K_{x}=\{g\in K:\, g.x=x\}}$.
In the following we analyze the topological structure of orbits $\mathcal{O}_{x}$.
In particular we invoke the Hopf-Samelson theorem \cite{Hopf Samelson}
which says that the Euler-Poincar\'{e} characteristics, $\chi(\mathcal{O}_{x})$
is always non-negative and positive if and only if the ranks of $K$
and $K_{x}$ are the same.

\subsection*{Euler-Poincar\'{e} characteristics}

The most fundamental topological invariant of a topological space
$X$ is the Euler-Poincar\'{e} characteristics. For a compact manifold
$X$ it is defined as

\begin{equation}
\mathcal{\chi}(X)=\sum_{p=0}^{\mathrm{dim}(M)}\left(-1\right)^{p}\mathrm{dim}C^{p}(X),
\end{equation}
where $C^{p}(K)$ denotes the space of closed $p$-forms. Recall that
the $p$-th cohomology group of $X$ is a quotient space of the closed
$p$-forms $C^{p}(X)$ by the exact $p$-forms $E^{p}(X)$

\[
H^{p}(X)=C^{p}(X)/E^{p}(X),
\]

\[
C^{p}(X)=\mathrm{Ker}\left(d:\bigwedge^{p}T^{\ast}X\rightarrow\bigwedge^{p+1}T^{\ast}X\right)\,,
\]

\[
E^{p}(X)=\mathrm{Ran}\left(d:\bigwedge^{p-1}T^{\ast}X\rightarrow\bigwedge^{p}T^{\ast}X\right).
\]
The beauty of the Euler-Poincar\'{e} characteristics is manifested
by the following equality

\begin{equation}
\sum_{p=0}^{\mathrm{dim}X}\left(-1\right)^{p}\mathrm{dim}C^{p}(X)=\mathcal{\chi}(X)=\sum_{p=0}^{\mathrm{dim}X}\left(-1\right)^{p}\mathrm{dim}H^{p}(X).\label{eq:euler}
\end{equation}

\subsection*{Euler-Poincar\'{e} characteristics of homogenous spaces }

Here we state the Hopf-Samelson theorem for the Euler-Poincar\'{e}
e characteristics of a homogenous space. We first recall concepts
of a maximal torus and of a Weyl group. A maximal torus of $K$, denoted
by $T$, is defined as a maximal abelian compact and connected subgroup
of $K$. A maximal torus is in general not unique. Yet, all maximal
tori are conjugate, i.e. they are of the form ${T'=gTg^{-1}}$ for ${g\in K}$.
The normalizer, ${N(T)}$ is defined by ${N(T)=\{g\in K:\, g^{-1}Tg=T\}}$.
Analogously, the centralizer, ${C(T)}$, is given by ${C(T)=\{g\in K:\, g^{-1}tg=t\,,t\in T\}}$.
Both ${N(T)}$ and ${C\left(T\right)}$ are compact. Centralizer ${C(T)}$
is clearly a normal subgroup of ${N(T)}$. Therefore the quotient group
${W_{T}=N(T)/C(T)}$ is well defined and finite \cite{Hall}. The group
${W_{K}}$ is called a Weyl group of $K$. For example, when ${K=SU(N)}$,
the maximal torus $T$ may be chosen to consists of diagonal matrices
in ${SU(N)}$. The corresponding $W_{K}$ is given by these matrices
in $K$ which permute diagonal elements of ${t\in T}$ and hence ${|W_{K}|=N!}$.
We can now state the Hopf-Samelson theorem.
\begin{fact}
(Hopf-Samelson theorem) \label{(Hopf-Samelson-theorem)}Let $K$ be
a compact connected Lie group and $K_{x}$ be a closed subgroup of
$K$ . If $K_{x}$ contains some maximal torus of $K$, ie. ${T\subset K_{x}}$,
the Euler-Poincar\'{e} characteristic is given by ${\chi(K/K_{x})=\frac{|W_{K}|}{|W_{K_{x}}|}}$.
When $K_{x}$ does not contain a maximal torus ${\chi(K/K_{x})=0}$.
\end{fact}

\section{Euler-Poincar\'{e} characteristic of orbits through pure separable,
$CC$, and $CQ$ states \label{sec:Topol_results}}

In this section we compute Euler-Poincar\'{e} characteristic for
orbits through pure separable, $CC$, and $CQ$ states. We show that
the Euler-Poincar\'{e} characteristic of the orbit of the relevant
group distinguishes classes of pure separable, $CC$, and $CQ$ states.
We start with the discussion of pure separable states. While discussing
$CC$ and $CQ$ states we will use the notation from Section \ref{sec:geometry of mixed states}.

\subsection{Pure separable states}

Consider the system consisting of $L$ distinguishable particles described
by 
\begin{equation}
{\mathcal{H}=\mathbb{C}^{N_{1}}\otimes\mathbb{C}^{N_{2}}\otimes\ldots\otimes\mathbb{C}^{N_{L}}}.
\end{equation}
Manifold of pure states, ${\mathbb{P}\mathcal{H}}$, consist of all
rank-one orthogonal projectors defined on $\mathcal{H}$. The group
of local unitary operations ${K=SU(N_{1})\times SU(N_{2})\times\ldots\times SU(N_{L})}$
acts on ${\mathbb{P}\mathcal{H}}$ in a natural manner,
\[
k.\kb{\psi}{\psi}=k\kb{\psi}{\psi}k^{-1},\,\text{for }k\in K\,\text{and }\kb{\psi}{\psi}\in\mathbb{P}\mathcal{H}\,.
\]
The manifold of separable states, $\mathrm{Sep}$, consists of states
of the form
\[
\kb{\psi}{\psi},\text{where }\ket{\psi}=\ket{\phi_{1}}\otimes\ket{\phi_{2}}\otimes\ldots\otimes\ket{\phi_{L}},
\]
for some normalized ${\ket{\psi_{i}}\in\mathbb{C}^{N_{i}}}$. Separable
states form the $K$-orbit through $\kb{\psi_{0}}{\psi_{0}}$, where
$\ket{\psi_{0}}$ is a simple tensor. In what follows we will prove
that the manifold of separable states is the only $K$-orbit that
have non-vanishing Euler-Poincar\'{e} characteristic. Any maximal
torus of $K$ is given by ${T=T_{1}\times T_{2}\times\ldots\times T_{L}}$,
where each $T_{i}$ is some maximal torus of ${SU(N_{i})}$. It is a
matter of straightforward calculation to see that

\begin{equation}
|W_{K}|=N_{1}!\cdot N_{2}!\cdot\ldots\cdot N_{L}!\,.
\end{equation}
One easily checks that if a stabilizer of a given state $K_{\kb{\psi}{\psi}}$
contains some $T$, then $\kb{\psi}{\psi}$ is separable (proof is
analogous to the proof of Proposition \ref{torus}). Moreover, it
is known that for the separable $\kb{\psi}{\psi}$ we have \cite{Zyczkowski Bengengson}

\begin{equation}
K_{\kb{\psi}{\psi}}=\prod_{i=1}^{i=L}S\left(U(1)\times SU(N_{i}-1)\right).
\end{equation}
Straightforward calculation gives
\begin{equation}
|W_{K_{\kb{\psi}{\psi}}}|=\left(N_{1}-1\right)!\cdot\left(N_{2}-1\right)!\cdot\ldots\cdot\left(N_{L}-1\right)!\,.
\end{equation}
Taking into account the above discussion and Fact \ref{(Hopf-Samelson-theorem)}
we arrive at the following Proposition.
\begin{prop}
Among all orbits of the local unitary group $K=SU(N_{1})\times SU(N_{2})\times\ldots\times SU(N_{L})$
in the manifold of pure states only the orbit of separable states
have the the non-vanishing Euler-Poincar\'{e} characteristic. It
is given by ${\chi\left(\mathrm{Sep}\right)=N_{1}\cdot N_{2}\cdot\ldots\cdot N_{L}}.$
\end{prop}

\subsection{CC and CQ states}

In order to compute the Euler-Poincar\'{e} characteristic of orbits
through $CC$ and $CQ$ states we need to compute the stabilizers
of the action of relevant groups acting on the state of interest.
We first show that $CC$ and $CQ$ states are uniquely characterized
by the non-vanishing Euler-Poincar\'{e} characteristic of orbits
${SU(N_{1})\times SU(N_{2})}$ and ${SU(N_{1})\times I_{N_{2}}}$ respectively.
\begin{prop}
\label{eurel cc char}Let ${\rho\in\mathcal{O}_{\rho_{0}}}$ be a mixed
state of a bipartite system ${\mathbb{C}^{N_{1}}\otimes\mathbb{C}^{N_{2}}}$.
Let $K=SU(N_{1})\times SU(N_{2})$ be a group of the local unitary
operations. The $K$-orbit through $\rho$, ${K.\rho}$, has non-vanishing
Euler-Poincar\'{e} characteristic if and only if $\rho$ is a $CC$
state. \end{prop}
\begin{proof}
Every maximal torus of $K$ is of the form ${T_{1}\times T_{2}}$, where
$T_{1}$ and $T_{2}$ are maximal tori of ${SU(N_{1})}$ and ${SU(N_{2})}$
respectively. By Proposition \ref{torus}, if $\rho$ is stabilized
by some $T$, then $\rho$ is a $CC$ state. Conversely, straightforward
calculation shows that every $CC$ state is stabilized by some maximal
torus ${T\subset K}$. By the Fact \ref{(Hopf-Samelson-theorem)} we
get that ${\chi\left(K.\rho\right)\neq0}$.\end{proof}
\begin{prop}
Let ${\rho\in\mathcal{O}_{\rho_{0}}}$ be a mixed state of a bipartite
system ${\mathbb{C}^{N_{1}}\otimes\mathbb{C}^{N_{2}}}$. Let $K=SU(N_{1})\times I_{N_{2}}$.
The $K$-orbit through $\rho$, ${K.\rho}$, has the non-vanishing Euler-Poincar\'{e}
characteristic if and only if $\rho$ is a $CQ$ state.\end{prop}
\begin{proof}
The proof is completely analogous to the proof of Proposition \ref{eurel cc char}.
\end{proof}
We next compute the Euler-Poincar\'{e} characteristic. We start with
the case of $CQ$ states as the computation for $CC$ will follow
from the former.
\begin{prop}
\label{euler cq} For a $CQ$ state $\rho$ and ${K=SU(N_{1})\times I_{N_{2}}}$,
under the notation used in Proposition \ref{main CQ}, the Euler characteristic
of ${K.\rho}$ is equal to
\begin{equation}
\chi(K.\rho)=\frac{N_{1}!}{\prod_{\sigma\in\mathcal{Q}}\left|\mathcal{I}_{\sigma}\right|!}.
\end{equation}
\end{prop}
\begin{proof}
Due to Fact \ref{(Hopf-Samelson-theorem)} , ${\chi(K.\rho)=\frac{|W(K)|}{|W(K_{\rho})|}}$.
Obviously ${|W(K)|=|W(SU(N_{1}))|=N_{1}!}$, so it suffices to find
${|W(K_{\rho})|}$. Because $\mathrm{Stab}_{K}\rho$ is connected (see
Appendix 2 for the proof), in order to find $K_{\rho}$ it is enough
to find the Lie algebra of the stabilizer, ${Lie\left(K_{\rho}\right)}$.
It is precisely the annihilator of $\rho$ with respect to the adjoint
action of $\mathfrak{k}$. Since all non-zero elements of $\{[E_{\alpha},\rho]\}$
are linearly independent (it follows directly from  (\ref{orbit-dim-1})),
the annihilator of $\rho$ is spanned by all $iH_{\alpha}$ and all
${i(E_{\alpha}+E_{-\alpha})},\,{(E_{\alpha}-E_{-\alpha})}$ for which
${[E_{\alpha},\rho]=0}$. Recall that ${[E_{ij}\otimes I_{N_{2}},\rho]=0}$
if and only if ${i,j\in\mathcal{I}_{\sigma}}$ for some ${\sigma\in\mathcal{Q}}$.
Therefore we have
\begin{equation}
Lie\left(K_{\rho}\right)=\left(\bigoplus_{\sigma\in\mathcal{Q}}\mathrm{Span}_{\mathbb{R}}(E_{ij}\otimes I_{N_{2}},\, H_{ij}\otimes I_{N_{2}}:i,j\in\mathcal{I}_{\sigma})\right)\oplus\tilde{\mathfrak{t}},
\end{equation}
where $\tilde{\mathfrak{t}}$ is spanned by elements of the Cartan
algebra that are not included in the first part of the above expression.
Note that

\[
\left(\bigoplus_{\sigma\in\mathcal{Q}}\mathrm{Span}_{\mathbb{R}}(E_{ij}\otimes I_{N_{1}},\, H_{ij}\otimes I_{N_{2}}:i,j\in\mathcal{I}_{\sigma})\right)\oplus\tilde{\mathfrak{t}}\approx\left(\bigoplus_{\sigma\in\mathcal{Q}}\mathfrak{su}(|\mathcal{I}_{\sigma}|)\right)\oplus\mathbb{R}^{k}\,,
\]
for some $k$. Passing to the Lie group picture we get

\begin{equation}
K_{\rho}\approx\left(\prod_{\sigma\in\mathcal{Q}}SU\left(|\mathcal{I}_{\sigma}|\right)\right)\times\mathbb{T}^{k}\,,\label{eq:stab cq}
\end{equation}
where ${\mathbb{T}^{k}=U(1)^{\times k}}$ is a $k$-dimensional torus.
It follows that
\[
\left|W\left(K_{\rho}\right)\right|=\left|W\left(\prod_{\sigma\in\mathcal{Q}}SU\left(|\mathcal{I}_{\sigma}|\right)\right)\right|=\prod_{\sigma\in\mathcal{Q}}\left|W\left(SU\left(|\mathcal{I}_{\sigma}|\right)\right)\right|=\prod_{\sigma\in\mathcal{Q}}|\mathcal{I}_{\sigma}|!\,.
\]

\end{proof}
\noindent We will use parts of the above reasoning in the computation
of ${\chi(K.\rho)}$ for $CC$ states.
\begin{prop}
\label{euler CC}For a $CC$ state $\rho$ and $ { K=SU(N_{1})\times SU(N_{2}) }$,
under the notation used Proposition \ref{mainCC}, the Euler-Poincar\'{e}
characteristic of ${K.\rho}$ is equal to
\begin{equation}
\chi(K.\rho)=\frac{N_{1}!}{\prod_{R\in\mathcal{R}}\left|\mathcal{I}_{R}\right|!}\frac{N_{2}!}{\prod_{C\in\mathcal{C}}\left|\mathcal{I}_{C}\right|!}.
\end{equation}
\end{prop}
\begin{proof}
Due to Fact \ref{(Hopf-Samelson-theorem)}, ${ \chi(K.\rho)=\frac{|W(K)|}{|W(K_{\rho})|} }$.
Obviously $ |W(K)|=|W(SU(N_{1})\times SU(N_{2}))|=N_{1}!N_{2}! $. In
order to compute $|W(K_{\rho})|$ we find $K_{\rho}$. Just like in
the $CQ$ case $K_{\rho}$ turns out to be connected (see Appendix 2
for the proof) and in order to find $K_{\rho}$ it is enough to find
its Lie algebra, that is, the annihilator of $\rho$ with respect
to the adjoint action of $\mathfrak{k}$. We have ${\mathfrak{k}=\mathfrak{su}(N_{1})\oplus\mathfrak{su}(N_{2})}$.
For ${X\in\mathfrak{su}(N_{1})}$ and ${Y\in\mathfrak{su}(N_{2})}$ the
non-zero elements of the form ${\left[X\otimes I_{N_{2}},\,\rho\right]}$
and ${\left[I_{N_{1}}\otimes Y,\,\rho\right]}$ are linearly independent
(see (\ref{eq:p-alpha-0}) and (\ref{eq:p-alpha-01})). Therefore,
the annihilator of $\rho$ with respect to the adjoint action of $\mathfrak{k}$
is a direct sum of annihilators with respect to actions of $\mathfrak{su}(N_{1})$
and $\mathfrak{su}(N_{2})$ taken separately:

\begin{eqnarray}
\mathfrak{su}(N_{1}) & \ni & X\rightarrow[X\otimes I_{N_{2}},\,\rho]\,,\label{eq:rep left}\\
\mathfrak{su}(N_{2}) & \ni & Y\rightarrow[I_{N_{1}}\otimes Y,\,\rho]\,.\label{eq:rep right}
\end{eqnarray}
From the perspective of the representations (\ref{eq:rep left}) and
(\ref{eq:rep right}) the state $\rho$ can be considered separately
as $CQ$ or $QC$ state. The definition of a $QC$ state is analogous
to the definition of $CQ$ state. Annihilators of $\rho$ with respect
to the action of $\mathfrak{su}(N_{1})$ or $\mathfrak{su}(N_{2})$
are thus annihilators of $\rho$ treated as a $CQ$ or $QC$ state.
Let $\mathcal{Q}^{(1)}$ and $I_{\sigma}^{(1)}$ be sets $\mathcal{Q}$
and $\mathcal{I}_{\sigma}$ (see  (\ref{eq:def Q}) and  (\ref{eq:def I sigma})
) when $\rho$ is treated as a $CQ$ state. Analogously, let $\mathcal{Q}^{(2)}$
and $I_{\sigma}^{(2)}$ be sets $\mathcal{Q}$ and $\mathcal{I}_{\sigma}$
when $\rho$ is treated as a $QC$ state. Repeating the reasoning
from the proof of Proposition \ref{euler cq} we get (see  (\ref{eq:stab cq}))

\[
K_{\rho}\approx\left(\prod_{\sigma\in\mathcal{Q}^{(1)}}SU\left(|\mathcal{I}_{\sigma}^{(1)}|\right)\right)\times\left(\prod_{\sigma\in\mathcal{Q}^{(2)}}SU\left(|\mathcal{I}_{\sigma}^{(2)}|\right)\right)\times\mathbb{T}^{k}\,,
\]
for some $k$. Therefore we have

\begin{equation}
\left|W\left(K_{\rho}\right)\right|=\left|W\left(\left(\prod_{\sigma\in\mathcal{Q}^{(1)}}SU\left(|\mathcal{I}_{\sigma}^{(1)}|\right)\right)\times\left(\prod_{\sigma\in\mathcal{Q}^{(2)}}SU\left(|\mathcal{I}_{\sigma}^{(2)}|\right)\right)\right)\right|\label{eq:final}
\end{equation}

\[
=\left(\prod_{\sigma\in\mathcal{Q}^{(1)}}|\mathcal{I}_{\sigma}^{(1)}|!\right)\cdot\left(\prod_{\sigma\in\mathcal{Q}^{(2)}}|\mathcal{I}_{\sigma}^{(2)}|!\right).
\]
We now consider different presentations of a $CC$ state ${\rho=\sum_{i,j}\,p_{ij}E_{ii}\otimes E_{jj}}$
that are suitable when it is treated as a $CQ$ or $QC$ state. We
define the marginal distributions ${p_{i}^{(1)}=\sum_{j}\,p_{ij}}$ and
${p_{j}^{(2)}=\sum_{i}\,p_{ij}}$. We have

\begin{equation}
\rho=\sum_{i}p_{i}^{(1)}E_{ii}\otimes\sigma_{i}^{(2)}\,,\label{eq:cq pres}
\end{equation}
where ${\sigma_{i}^{(2)}=\sum_{j}\,\frac{p_{ij}}{p_{i}^{(1)}}E_{jj}}$.
Expression (\ref{eq:cq pres}) is useful when $\rho$ is treated as
a $CQ$ state. Analogously we have

\begin{equation}
\rho=\sum_{j}p_{j}^{(2)}\sigma_{j}^{(1)}\otimes E_{jj}\,,\label{eq:qc pres}
\end{equation}
where ${\sigma_{j}^{(1)}=\sum_{i}\,\frac{p_{ij}}{p_{j}^{(2)}}E_{ii}}$.
On the other hand, expression (\ref{eq:cq pres}) is useful when $\rho$
is treated as a $QC$ state.  Closer examination of (\ref{eq:cq pres})
shows that the set $\mathcal{Q}^{(1)}$ is in the bijection with the
$\mathcal{C}$. That is each ${\sigma\in\mathcal{Q}^{(1)}}$ correspond
to the unique ${C\in\mathcal{C}}$. For such a pair we have ${\left|\mathcal{I}_{\sigma}^{(1)}\right|=\left|\mathcal{J}_{C}\right|}$.
Similarly, we have the bijection between $\mathcal{Q}^{(2)}$ and
$\mathcal{R}$. Each $\sigma\in\mathcal{Q}^{(2)}$ correspond to the
unique $R\in\mathcal{R}$ and we have ${\left|\mathcal{I}_{\sigma}^{(2)}\right|=\left|\mathcal{I}_{R}\right|}$.
These two observations together with (\ref{eq:final}) conclude the
proof.
\end{proof}

\section{Summary and outlook \label{sec:Concluding-remarks-and}}

We have showed that geometric and topological methods can be applied
to distinguish interesting classes of mixed states of composite quantum
systems. On the geometrical side we proved that for bipartite system
${\mathbb{C}^{N_{1}}\otimes\mathbb{C}^{N_{2}}}$ the generic $CC$ states
are distinguished as symplectic orbits of ${SU(N_{1})\times SU(N_{2})}$
in the manifold of isospectral density matrices. Similarly, the generic
$CQ$ states are distinguished as symplectic orbits of ${SU(N_{1})\times I_{N_{2}}}$
in the same manifold. On the topological side we studied the Euler-Poincar\'{e}
characteristic of orbits of relevant groups through arbitrary multipartite
pure, $CC$, and $CQ$ states. We proved that non-zero Euler-Poincar\'{e}
characteristic of orbits of the local unitary group through pure multipartite
and bipartite mixed states characterizes pure separable and $CC$
states. Analogously, non-vanishing Euler-Poincar\'{e} characteristic
of orbits ${SU(N_{1})\times I_{N_{2}}}$ on bipartite mixed states detects
$CQ$ states. Above results can be easily generalized to mixed states
of multipartite systems. For example in the tripartite case, geometric
and topological features of the orbits of the suitably chosen groups
should distinguish classes of $CCC$ or $CCQ$ states. Another interesting
generalization would involve the usage of the same geometric and topological
methods to study mixed states of fermionic and bosonic systems.

\section{Acknowledgments}

We would like to thank Marek Ku\'{s} and Jaros\l{}aw Korbicz for fruitful
discussions. {We are grateful to the anonymous reviewer for pointing out mistakes in our original treatment of  the problem of induced K\"{a}hler structures}.  We acknowledge the support of Polish MNiSW
Iuventus Plus grant no. IP2011048471,SFB/TR12 Symmetries and Universality
in Mesoscopic Systems program of the Deutsche Forschungsgemeischaft
and a grant of the Polish National Science Centre under the contract
number DEC-2011/01/M/ST2/00379 .

\section{Appendix 1: Complex and almost-complex structures}
{
In this Appendix  we show that the almost integrable structure on the symplectic orbit ${K.\rho \subset \mathcal{O}_{\rho_0}}$ introduced in the proof of Proposition \ref{complex orbits} is integrable. In order to do it we recall (see [\onlinecite{nakahara}] for a more detailed treatment of complex and almost complex structures) one of the equivalent definitions of integrability of the complex structure $J$. Let ${T\mathcal{M}^{\mathbb{C}}}$ be the complexified tangent boundle of $\mathcal{M}$. For each ${x\in\mathcal{M}}$ the complex structure $J$ defines the decomposition of ${T_x \mathcal{M}^{\mathbb{C}}}$,
\begin{equation}\label{splitting}
T_x \mathcal{M}^{\mathbb{C}}=T^{+}_x \mathcal{M}^{\mathbb{C}} \oplus T^{-}_x \mathcal{M}^{\mathbb{C}}\ ,
\end{equation}
where
\[
T^{+}_x = \mathcal{M}^{\mathbb{C}}=\mathrm{ker}\left(J_x -i \mathbb{I}_x \right)\ ,\ T^{-}_x= \mathcal{M}^{\mathbb{C}}=\mathrm{ker}\left(J_x +i \mathbb{I}_x \right) \ .
\]
In the above expression $\mathbb{I}_x$ denotes the identity operator on the complexified tangent space at $x$, ${T_x \mathcal{M}^{\mathbb{C}}}$. Vectors belonging to ${T^{+}_x \mathcal{M}^{\mathbb{C}}}$ and ${T^{-}_x \mathcal{M}^{\mathbb{C}}}$ are called holonomic and respectively anti-holonomic. Introducting the notation.
\[
T^{+} \mathcal{M}^{\mathbb{C}}=\bigcup_{x\in\mathcal{M}} T^{+}_x \mathcal{M}\ ,\ T^{-} \mathcal{M}^{\mathbb{C}}=\bigcup_{x\in\mathcal{M}} T^{-}_x \mathcal{M} \ ,
\]
 we get the decomposition of the complexified tangent bundle,
\begin{equation}
T\mathcal{M}^{\mathbb{C}}= T^{+} \mathcal{M}^{\mathbb{C}} \oplus T^{-} \mathcal{M}^{\mathbb{C}}\ .
\end{equation}
Sections of bundles ${T^{+} \mathcal{M}^{\mathbb{C}}}$ and ${T^{-} \mathcal{M}^{\mathbb{C}}}$ are called holonomic and respectively anti-holonomic vector fields. We denote the (complex) vector spaces of holonomic and anti-holonomic vector fields by ${\mathcal{X}\left(\mathcal{M}\right)^+}$ and ${\mathcal{X}\left(\mathcal{M}\right)^-}$ respectively. We are now ready to give the definition of integrability of the almost complex structure $J$. Almost complex structure $J$ is integrable\cite{nakahara} if and only if for every ${V,W\in\mathcal{X}\left(\mathcal{M}\right)^+}$ we have
\begin{equation}\label{integrab}
\left[V,\ W\right]_{\mathrm{Lie}}\in\mathcal{X}\left(\mathcal{M}\right)^+ \ ,
\end{equation}
where ${\left[\cdot,\ \cdot\right]_{\mathrm{Lie}}}$ denotes the Lie bracket of vector fields.
We now apply the notions introduced above to prove the integrability of the complex structure introduced in Proposition \ref{complex orbits}. We have ${\mathcal{M}=K.\rho}$. We assume ${K.\rho}$ is symplectic i.e. that it inherits the symplectic structure $\omega$ from $\mathcal{O}_{\rho_0}$. As before, the symplectic structure on ${K.\rho}$ is denoted by $\left.\omega\right|_{K.\rho}$.
By $\left.J\right|_{K.\rho}$ we denote the associated almost complex structure. Because symplectic structure $\left.\omega\right|_{K.\rho}$ is $K$ invariant is is enough the check the integrability of the associated almost complex structure in the neighbourhood of $\tilde{\rho}$ such that ${\left[\tilde{\rho},\ X\right]}$ for ${X\in\mathfrak{t}}$. Let $\mathcal{U}$ denote such a neighbourhood.  Because ${K.\rho}$ is symplectic for ${x=k\tilde{\rho}k^{-1}\in\mathcal{U}}$, ${k\in K}$  we have the decomposition
\begin{equation}
T_{x}K.\rho= \bigoplus_{\alpha\in\mathcal{A}}P_\alpha \left(x\right) \ \label{trivialisation},
\end{equation}
where
\[
P_\alpha \left(x\right)=\mathrm{Span}_{\mathbb{R}}\left( k\left [E_{\alpha}-E_{-\alpha},\,\tilde{\rho}\right] k^{-1},\, i k\left[E_{\alpha}-E_{-\alpha},\,\tilde{\rho}\right] k^{-1}\right)\ ,\label{eq:decop p(x)}
\]
and
\[
\mathcal{A}=\left\{ \alpha \middle|\,\alpha>0\ ,\mathrm{tr}\left(\tilde{\rho}H_{\alpha}\right)\neq0\right\} \ .
\]

Direct computation (see \eqref{eq:adrho}) shows that the decomposition of the complexified tangent space at $x$,${TK.\rho^\mathbb{C}}$  reads as
\[
T_x^+K.\rho^\mathbb{C}=  \mathrm{Span}_{\mathbb{C}}\left( k\left [E_{\alpha},\,\tilde{\rho}\right] k^{-1} |\  \alpha\in\mathcal{A}\right)  ,\ T_x^-K.\rho^\mathbb{C}= \mathrm{Span}_{\mathbb{C}}\left( k\left [E_{-\alpha},\,\tilde{\rho}\right] k^{-1} |\  \alpha\in\mathcal{A}\right)  \ .
\]
We observe that \eqref{trivialisation} gives the local trivialistation of ${T\mathcal{U}^\mathbb{C}}$, ${TU^{\mathbb{C}}\approx \mathcal{U} \times \mathfrak{k}^\mathbb{C}}$. Using this and the commutation relations in the Lie algebra $\mathfrak{k}^\mathbb{C}$  we conclude that for the above defined almost complex structure condition \eqref{integrab} is satisfied and therefore it is integrable.

}
\section*{Appendix 2: Stabilizes of orbits through CC and CQ states}

In this part we show that the stabilizers of the action of the relevant
groups through $CC$ and $CQ$ states are connected. This result was
used in the proofs of Propositions \ref{euler cq} and \ref{euler CC}.
Let us first recall that that a (not necessary maximal) torus $T$
of a Lie group $G$ is defined as a compact connected abelian subgroup
of $G$. We have the following fact.
\begin{fact}
\cite{Compact groups} \label{stab tori}Let ${T\subset K}$ be a torus
in a connected compact Lie group $K$. Let
\begin{equation}
C\left(T\right)=\left\{ k\in K \middle|kt=tk,\, t\in T\right\},
\end{equation}
be a centralizer of $T$ in $K$. Then the centralizer ${C\left(T\right)}$
is a connected compact subgroup of $K.$
\end{fact}
\noindent  We are now ready to present proofs.
\begin{prop}
\label{connected cc}Let $\rho$ be a $CC$ state of a bipartite system
${\mathbb{C}^{N_{1}}\otimes\mathbb{C}^{N_{2}}}$. Let ${K=SU(N_{1})\times SU(N_{2})}$.
Then the stabilizer of $\rho$, ${K_{\rho}=\left\{ k\in K \middle|\, k.\rho=\rho\right\}} $
, is a connected subgroup of $K$.\end{prop}
\begin{proof}
Any $CC$ state can be written in the form

\begin{equation}
\rho=\sum_{l}p_{l}\mathbb{P}_{l},\label{eq:cc eigen}
\end{equation}
where summation is over index labeling distinct eigenvalues $p_{l}$
and $\mathbb{P}_{k}$ are projectors onto eigenspaces of $\rho$.
Note that $\mathbb{P}_{l}$ are formed from rank 1 projectors onto
separable tensors. Therefore, to the decomposition (\ref{eq:cc eigen})
we can associate a unique torus ${T\subset K}$

\[
T=\left\{ \sum_{l}e^{i\phi_{l}}\mathbb{P}_{l}\middle|\sum_{l}\phi_{l}=0\right\} .
\]
One easily checks that ${C(T)=K_{\rho}}$ and therefore, by Fact \ref{stab tori},
$K_{\rho}$ is connected.\end{proof}
\begin{prop}
Let $\rho$ be a $CQ$ state of a bipartite system ${\mathbb{C}^{N_{1}}\otimes\mathbb{C}^{N_{2}}}$.
Let ${K=SU(N_{1})\times I_{N_{2}}}$. Then the stabilizer of $\rho$,
${K_{\rho}=\left\{ k\in K \middle|\, k.\rho=\rho\right\}} $ , is a connected
subgroup of $K$.\end{prop}
\begin{proof}
Any $CQ$ state can be written in the form ( see (\ref{eq:def Q}))

\begin{equation}
\rho=\sum_{\sigma\in\mathcal{Q}}\mathbb{P}_{\sigma}\otimes\sigma,
\end{equation}
where ${\mathbb{P}_{\sigma}=\sum_{i\in\mathcal{I}_{\sigma}}\kb ii}$.
Using the fact that ${k\in SU(N_{1})\times I_{N_{2}}}$ stabilizes $\rho$
if and only if $k$ preserves eigenspaces of $\rho$ and repeating
the argument from the proof of Proposition \ref{connected cc} we
get that ${K_{\rho}=C(T)}$ where $T$ is a torus in $K$ given by

\[
T=\left\{ \sum_{\sigma\in\mathcal{Q}}e^{i\phi_{\sigma}}\mathbb{P}_{\sigma}\otimes I_{N_{2}} \middle|\sum_{\sigma\in\mathcal{Q}}\phi_{\sigma}=0\right\} .
\]
Hence, by Fact \ref{stab tori}, $K_{\rho}$ is connected.
\end{proof}


\begin{thebibliography}{References}
\bibitem{Discord Zurek}H.Ollivier and W. H. Zurek, "Quantum Discord:
A Measure of the Quantumness of Correlations", \textit{Phys. Rev. Lett.}
\textbf{88} 017901 (2001)

\bibitem{Discord Verdal}L. Henderson L and V. Vedral, "Classical,
quantum and total correlations", \textsl{J. Phys A} \textbf{34} 6899 (2001)

\bibitem{Nullity of quantum discord}A. Datta, "Condition for the
Nullity of Quantum Discord", (\textsl{Preprint} aXiv:1003.5256)

\bibitem{Horo 1}J. Oppenheim, M. Horodecki, P. Horodecki, and R. Horodecki, "Thermodynamical Approach to Quantifying Quantum Correlations"
\textsl{Phys. Rev. Lett.} \textbf{89} 180402 (2002)

\bibitem{Horo 2}M. Piani, P. Horodecki, and R. Horodecki,  "No-Local-Broadcasting
Theorem for Multipartite Quantum Correlations", \textsl{Phys. Rev. Lett.}
\textbf{100} 090502 (2008)

\bibitem{Korbicz CC and CQ}J. K. Korbicz, P. Horodecki, R. Horodecki, "Quantum-correlation breaking channels, broadcasting scenarios,
and finite Markov chains", \textsl{Phys. Rev. A} \textbf{86} 042319 (2012)

\bibitem{chusc} D. Chru\'sci\'nski, "Quantum-correlation breaking
channels, quantum conditional probability and Perron\textendash{}Frobenius
theory", \textsl{Phys. Lett. A} \textbf{377} 606 (2013)



\bibitem{Acin almost all} A. Ferraro A, L. Aolita, D. Cavalcanti, F. M. Cucchietti
, and A. Acin, "Almost all quantum states have nonclassical correlations",
\textsl{Phys. Rev. A} \textbf{81} 052318 (2010)

\bibitem{Symplectic geometry of entanglement} A. Sawicki, A. Huckleberry, and M. Ku\'s, "Symplectic geometry of entanglement",  \textsl{ Comm. Math.
Phys. }\textbf{305} 441 (2011)

\bibitem{HKS12} A. Huckleberry A, M. Ku\'s and A. Sawicki, "Bipartite
entanglement, spherical actions, and geometry of local unitary orbits",
\textsl{J. Math. Phys.} \textbf{54} 022202 (2013)

\bibitem{Hopf Samelson}H. Hopf and H. Samelson, "Ein Satz uber die
Wirkungsraume geschlossener Liescher Gruppen", \textsl{Comment. Math.
Helv.} \textbf{13} 240 (1941, In German)

\bibitem{Rieffel}M. A. Rieffel, "Dirac operators for coadjoint orbits
of compact Lie groups", \textsl{Munster J. of Math.} \textbf{2} 265 (2009)

\bibitem{sternberg}V. Guillemin V and S. Sternberg, \textsl{Symplectic
techinques in Physics}, (Cambridge: Cambridge University Press 1984)

\bibitem{Hall} B. C. Hall,  \textsl{Lie groups, Lie algebras, and
representations: an elementary introduction} (New York: Springer 2003)

\bibitem{Chruscinski}D. Chru\'scinski, "Symplectic orbits in quantum
state space", \textsl{J. Math. Phys.} \textbf{31} 1587 (1990)

\bibitem{bengsson} A. Baecklund, I. Bengsson, Four remarks on spin coherent states, \url{arXiv:1312.2427}


\bibitem{rectangles}T. Hagedorn, "On the existence of magic n-dimensional
rectangles", \textsl{Discrete Math.} \textbf{207} 53 (1999)

\bibitem{Zyczkowski Bengengson}I. Bengtsson and K. \.Zyczkowski
\textsl{Geometry of Quantum States} (Cambridge: Cambridge University
Press 2006)

\bibitem{Compact groups}K. F. Hofmann and S. A Morris, \textsl{The
Structure of Compact Groups: A Primer for Students - a Handbook for
the Expert} (Berlin: de Gruyter Studies in Mathematics 2006)

\bibitem{nakahara} Mikio Nakahara, \textsl{Geometry, Topology and Physics, Second Edition} (IOP Publishing 2003) 
 \end{thebibliography}
\end{document}